%% file: main.tex
\documentclass[conference]{IEEEtran}

\usepackage{balance} %% balance the columns on the last page
\usepackage{cite}
\usepackage{amsmath,amssymb,amsfonts}
\input{setup}
\title{Compositional Non-Interference \\ for Fine-Grained Concurrent Programs}

\author{
  \IEEEauthorblockN{Dan Frumin}
  \IEEEauthorblockA{Radboud University} \and
  \IEEEauthorblockN{Robbert Krebbers}
  \IEEEauthorblockA{Delft University of Technology}\and
  \IEEEauthorblockN{Lars Birkedal}
  \IEEEauthorblockA{Aarhus University}}

\begin{document}
\maketitle
% TODO: disable page numbers in final version
\thispagestyle{plain}
\pagestyle{plain}

\begin{abstract}
\input{abstract}

\end{abstract}


\input{intro}
\input{examples}
\input{preliminaries}
\input{logic}
\input{types}
\input{modular}

\input{soundness}
\input{formalization}
\input{related_work}
\input{conclusion}

\bibliography{ni,iris}

\appendices
\crefalias{section}{appendix}
\crefalias{subsection}{appendix}
\input{appendix}

\end{document}

%% file: setup.tex
\makeatletter%
\@ifundefined{basedir}{%
\newcommand\basedir{}%
}{}%
\makeatother%
\usepackage[draft,inline,nomargin]{fixme}

\usepackage{amsfonts}
\usepackage{amsthm}
\usepackage{amssymb}
\usepackage{stmaryrd}

\usepackage{mathpartir}

\usepackage{\basedir pftools}
\usepackage{\basedir iris}

\usepackage{xcolor}  % for print version

\usepackage{graphicx}
\usepackage[inline]{enumitem}
\usepackage{semantic}
\usepackage{csquotes}

\usepackage{hyperref}
\usepackage{cleveref}
\crefname{appendix}{appendix}{appendices}
\Crefname{appendix}{Appendix}{Appendices}
\usepackage{xspace}

\usepackage[english]{babel}
\usepackage[babel=true]{microtype}
\usepackage[all]{xy}
\usepackage{tikz}
\usetikzlibrary{automata, positioning, arrows}
\tikzset{->,  % makes the edges directed
  >=stealth', % makes the arrow heads bold
  node distance=3cm, % specifies the minimum distance between two nodes
  state/.style={
    rectangle,
    rounded corners,
    draw=black, thick,
    minimum height=1.7em,
    inner sep=2pt,
    text centered,
  },
}

\makeatletter
\newcommand{\crefnames}[3]{%
  \@for\next:=#1\do{%
    \expandafter\crefname\expandafter{\next}{#2}{#3}%
  }%
}
\makeatother

\makeatletter
\newcommand{\chlabel}[1]{%
  \def\@currentlabel{#1}%
  \phantomsection%
}
\makeatother

\crefnames{part,chapter,section}{\S\!}{\S}

\SetSymbolFont{stmry}{bold}{U}{stmry}{m}{n} % this fixes warnings when \boldsymbol is used with stmaryrd included

\extrarowheight=\jot	% else, arrays are scrunched compared to, say, aligned
\newcolumntype{.}{@{}}
\newcolumntype{M}{@{\mskip\thickmuskip}}

\definecolor{StringRed}{rgb}{.637,0.082,0.082}
\definecolor{CommentGreen}{rgb}{0.0,0.55,0.3}
\definecolor{KeywordBlue}{rgb}{0.0,0.3,0.55}
\definecolor{LinkColor}{rgb}{0.55,0.0,0.3}
\definecolor{CiteColor}{rgb}{0.55,0.0,0.3}
\definecolor{HighlightColor}{rgb}{0.0,0.0,0.0}

\definecolor{grey}{rgb}{0.5,0.5,0.5}
\definecolor{red}{rgb}{1,0,0}

\hypersetup{%
  linktocpage=true, pdfstartview=FitV,
  breaklinks=true, pageanchor=true, pdfpagemode=UseOutlines,
  plainpages=false, bookmarksnumbered, bookmarksopen=true, bookmarksopenlevel=3,
  hypertexnames=true, pdfhighlight=/O,
  colorlinks=true,linkcolor=LinkColor,citecolor=CiteColor,
  urlcolor=LinkColor
}

\theoremstyle{definition}
\newtheorem{proposition}{Proposition}
\newtheorem{definition}[proposition]{Definition}

\newtheorem{lemma}[proposition]{Lemma}
\newtheorem{theorem}[proposition]{Theorem}

\newcommand*{\Sref}[1]{\hyperref[#1]{\S\ref*{#1}}}
\newcommand*{\secref}[1]{\hyperref[#1]{Section~\ref*{#1}}}
\newcommand*{\lemref}[1]{\hyperref[#1]{Lemma~\ref*{#1}}}
\newcommand*{\thmref}[1]{\hyperref[#1]{Theorem~\ref*{#1}}}
\newcommand{\corref}[1]{\hyperref[#1]{Cor.~\ref*{#1}}}
\newcommand*{\defref}[1]{\hyperref[#1]{Definition~\ref*{#1}}}
\newcommand*{\egref}[1]{\hyperref[#1]{Example~\ref*{#1}}}
\newcommand*{\appendixref}[1]{\hyperref[#1]{Appendix~\ref*{#1}}}
\newcommand*{\figref}[1]{\hyperref[#1]{Figure~\ref*{#1}}}
\newcommand*{\tabref}[1]{\hyperref[#1]{Table~\ref*{#1}}}

\newcommand{\ie}{\emph{i.e.,} }
\newcommand{\cf}{\emph{cf.} }
\newcommand{\eg}{\emph{e.g.,} }
\newcommand{\etal}{\emph{et~al.}\xspace}
\newcommand{\wrt}{w.r.t.~}
\newcommand{\resp}{resp.~}
\newcommand{\etc}{\emph{etc.}}
\newcommand{\loccit}{\emph{loc. cit.} }

\fxsetup{theme=color,mode=multiuser}
\FXRegisterAuthor{dan}{adan}{\colorbox{green!70}{\color{black}Dan}}
\FXRegisterAuthor{robbert}{arobbert}{\colorbox{orange!70}{\color{black}Robbert}}
\FXRegisterAuthor{lars}{alars}{\colorbox{blue!70}{\color{white}Lars}}

\newcommand{\HeapLang}{\textlog{HeapLang}\xspace}
\newcommand{\Type}{\textlog{Type}}

\newcommand{\langkwStyle}[1]{\textbf{\color{blue} #1}}
\reservestyle{\langkw}{\langkwStyle}
\reservestyle{\langconst}{\textbf}
\reservestyle{\langother}{\mathit}
\langkw{let[let\spac],in[\spac in \spac],
  match,with,
  if[if\spac],then[then\spac],else[else\spac],
  ref,rec[rec\spac],fork,CAS,FAA,val,type}

\def\Match#1with#2{\<match> \spac #1 \spac \<with> \spac #2}
\def\Let#1=#2in{\<let> #1 \mathrel{=} #2 \<in> }
\def\If#1then{\<if> #1 \spac \<then>}
\def\Else{\spac \<else>}
\def\Ref(#1){\<ref>(#1)}
\def\Rec#1 #2={\<rec> {#1} \  {#2} \mathrel{=} }
\newcommand*\Fork[1]{\<fork>\spac\set{#1}}
\newcommand\deref{\mathop{!}}
\let\gets\leftarrow

\langconst{assert,false,true,inl[inl \spac],inr[inr \spac],None,Some,sig,end}

\newcommand{\apar}[2]{#1 \mathrel{||} #2}

\newcommand{\Ectx}{\textdom{ECtx}}
\newcommand\bigcdot{\mathrel{\raisebox{1pt}{$\scriptscriptstyle\bullet$}}}
\newcommand\holed[1]{[\,#1\,]}
\newcommand\exphole{\holed\bigcdot}
\newcommand\contextE[1]{\kern1pt\holed{#1}}

\newcommand{\pstep}[1][]{\xrightarrow{#1}_{\mathsf{tp}}}

\newcommand{\loc}{\ell}

\newcommand{\flatType}{\ensuremath{\mathrm{flat}}}
\newcommand{\pureexec}[2]{#1 \ra_{\mathsf{pure}} #2}

\newcommand{\typed}[3]{#1 \vdash #2 : #3}
\newcommand{\typedS}[3]{#1 \models #2 : #3}

\newcommand{\thelogic}{{SeLoC}\xspace}

\newcommand{\staterel}{\mathit{SR}}
\newcommand{\outloc}{\mathit{out}}
\newcommand{\Llocset}{\mathfrak{L}}
\newcommand{\lowEquiv}{\mathrel{\sim_\Llocset}}

\newcommand{\Slevel}{\textlog{Lbl}}
\newcommand{\lvlvar}{\chi}
\newcommand{\lvlvarB}{\xi}
\newcommand{\low}{\mathbf{L}}
\newcommand{\high}{\mathbf{H}}

\newcommand{\LRvar}{\theta}

\newcommand{\mapstoLR}{\mapsto_{\LRvar}}

\newcommand{\mapstoL}{\mapsto_\mathsf{L}}
\newcommand{\mapstoR}{\mapsto_\mathsf{R}}

\newcommand{\tint}{\mathtt{int}}
\newcommand{\tunit}{\mathtt{unit}}
\newcommand{\tref}[1]{\mathtt{ref}\ #1}
\newcommand{\tbool}{\mathtt{bool}}

\newcommand{\dwp}{\textlog{dwp}}
\NewDocumentCommand\dwpre{O{} m m O{} m}%
  {\dwp^{#1}_{#4}\spac#2\spac\&\spac#3\spac{\left\{#5\right\}}}

\newcommand{\awpsymb}{\textlog{awp}_{\LRvar}}
\NewDocumentCommand\awp{O{} m m}%
  {\textlog{awp}_{#1}\spac#2\spac{\left\{#3\right\}}}

\newcommand{\interp}[1]{\llbracket #1 \rrbracket}
\newcommand{\interpE}[1]{\llbracket #1 \rrbracket^\mathsf{e}}

\NewDocumentCommand \future {O{\top} O{\top}} {\later\kern-2.4ex\pvs[#1][#2]\kern-2.4ex\later\,}

\langother{rand}

\newcommand{\lkvar}{\mathit{lk}}
\newcommand{\acquire}{\mathit{acquire}}
\newcommand{\newlock}{\mathit{new\_lock}}
\newcommand{\release}{\mathit{release}}
\newcommand{\isLock}[3]{\ensuremath{\textlog{isLock}(#1, #2, #3)}}
\newcommand{\locked}[2]{\ensuremath{\textlog{locked}(#1, #2)}}

\newcommand{\arrayMake}[2]{\mathit{new\_array}\ #1\ #2}
\newcommand{\arrayGet}[2]{\mathit{array\_get}\ #1\ #2}

\newcommand{\STSState}{\textlog{State}}
\newcommand{\Classified}{\textlog{Classified}}
\newcommand{\Intermediate}{\textlog{Intermediate}}
\newcommand{\Declassified}{\textlog{Declassified}}

\newcommand{\inState}[1]{\textlog{in\_state}_{\gname}\ifthenelse{\equal{#1}{}}{}{(#1)}}
\newcommand{\stateToken}[1]{\textlog{state\_token}_{\gname}\ifthenelse{\equal{#1}{}}{}{(#1)}}

\newcommand{\newRec}{\mathit{new\_vdep}}
\newcommand{\classify}{\mathit{classify}}
\newcommand{\declassify}{\mathit{declassify}}
\newcommand{\readRec}{\mathit{read}}
\newcommand{\storeRec}{\mathit{store}}
\newcommand{\isClassified}{\mathit{get\_classified}}
\newcommand{\valueDependent}{\textlog{val\_dep}}
\newcommand{\classification}[2]{\textlog{class}_{(r_1, r_2)}(#1, #2)}
\newcommand{\classificationAuth}[1]{\textlog{class\_auth}_{(r_1, r_2)}(#1)}

\newcommand{\RR}{\mathrel{\mathcal{R}}}
\newcommand{\RRR}{\RR^{\ast}}

\renewcommand{\Prop}{\textlog{Prop}}

%% file: abstract.tex
Non-interference is a program property that ensures the absence of information leaks.
In the context of programming languages, there exist two common approaches for establishing non-interference: type systems and program logics.
Type systems provide strong automation (by means of type checking), but they are inherently restrictive in the kind of programs they support.
Program logics support challenging programs, but they typically require significant human assistance, and cannot handle modules or higher-order programs.

To connect these two approaches, we present \thelogic---a separation logic for non-interference, on top of which we build a type system using the technique of logical relations.
By building a type system on top of separation logic, we can compositionally verify programs that consist of typed and untyped parts.
The former parts are verified through type checking, while the latter parts are verified through manual proof.

The core technical contribution of \thelogic is a relational form of weakest preconditions that can track information flow using separation logic resources.
\thelogic is fully machine-checked, and built on top of the Iris framework for concurrent separation logic in Coq.
The integration with Iris provides seamless support for fine-grained concurrency, which was beyond the reach of prior type systems and program logics for non-interference.
%All our results are mechanized in the Coq proof assistant.

%% file: intro.tex
\section{Introduction}
\label{sec:intro}

\emph{Non-interference} is a form of \emph{information flow control}~(IFC) used to express
that confidential information cannot leak to attackers.
To establish non-interference of modern programs, it is crucial to develop verification techniques that support challenging programming paradigms and programming constructs such as concurrency.
Furthermore, to scale up these techniques to larger programs, it is important that they are compositional.
That is, they should make it possible to establish non-interference of program modules in isolation, without having to consider all possible interference from the environment and other program modules.

Much effort has been put into developing these techniques.
In terms of expressivity, techniques have been developed that support dynamically allocated references and higher-order functions~~\cite{flowcaml,RajaniGargModel,ZdancewicThesis}, and concurrency~\cite{assumptionsguarantees,prob-ni,CovernMain,CovernValueDep,CovernVeronica,comp-ni-sep,SecCSL}.
Despite recent advancements, the expressivity of available techniques for non-interference still lags behind the expressivity of techniques for functional correctness, which have seen major breakthroughs since the seminal development of concurrent separation logic~\cite{OHEARN2007,brookes:csl}.
There are several reasons for this.

First, a lot of prior work on non-interference focused on type systems and type system-like logics, \eg \cite{flowcaml,assumptionsguarantees,CovernMain,SecCSL,comp-ni-sep}.
  Such systems provide strong automation (by means of type checking), but lack capabilities to reason about functional correctness, and are thus inherently restrictive in the kind of programs they can verify.
For example, it may be the case that the confidentiality of the contents of a reference depends on runtime information instead of solely static information (this is called \emph{value-dependent classification} \cite{DynSecLabels,CovernValueDep,RHTT,DepIFTypes,GregersenIdris}).

Second, proving non-interference is harder than proving functional correctness.
  While functional correctness is a property about each single run of a program, non-interference is stated in terms of multiple runs of the same program.
  One has to show that for different values of confidential inputs, the attacker cannot observe a different behavior.

To overcome the aforementioned shortcomings, we take a new approach that combines program logics and type systems: we present a concurrent separation logic for non-interference on top of which we build a type system for non-interference.
Program modules whose non-interference relies on functional correctness (and thus cannot be type checked) can be assigned a type through manual proof by dropping down to our logic.
This combination of separation logic and type checking makes it possible compositionally to establish non-interference of programs that consist of untyped and typed parts.

Although ideas from concurrent separation logic have been employed in the context of non-interference before~\cite{comp-ni-sep,SecCSL}, we believe that in the context of non-interference the combination of typing and separation logic is new.
Moreover, our approach provides a number of other advantages compared to prior work:
\begin{itemize}
\item Our separation logic supports \emph{fine-grained concurrency}.
  That is, it can verify programs that use low-level atomic operations like compare-and-set to implement lock-free concurrent data structures and high-level synchronization mechanisms such as locks/mutexes.
  In prior work, such mechanisms were taken to be language primitives.
\item Our separation logic is \emph{higher-order}, making it possible to assign very general specifications to program modules.
\item Our separation logic is \emph{relational}, making it possible to reason about multiple runs of a program with different values for confidential inputs.
\item Our separation logic provides a powerful \emph{invariant} mechanism to describe protocols, making it possible to reason about sophisticated forms of sharing, as in value-dependent classifications.
\end{itemize}

In order to build our logic we make use of the Iris framework for concurrent separation logic~\cite{iris1,iris2,iris3,irisJFP}, which provides basic building blocks, including the invariant mechanism.
To combine typing and separation logic, we follow recent work on \emph{logical relations} in Iris \cite{irisIPM,irisEffects,irisRunSt,reloc,RustBelt}, but apply it to non-interference instead of functional correctness or contextual refinement.

\smallskip
\noindent \textbf{Contributions.}
We introduce \textbf{\thelogic}, the first separation logic for non-interference that combines typing and manual proof.

\begin{itemize}
\item We present a number of challenging examples that can be verified using \thelogic (\Cref{sec:examples}).
\item \thelogic supports a language with fine-grained concurrency, higher-order functions, and dynamic (higher-order) references (\Cref{sec:lang}).
   \thelogic is sound \wrt a standard timing-sensitive notion of non-interference---\emph{strong low-bisimulations}---by Sabelfeld and Sands \cite{prob-ni} (\Cref{sec:security}).
\item To verify challenging programs, \thelogic features a relational version of weakest preconditions, which integrates seamlessly with the powerful mechanism for invariants and protocols of the Iris framework (\Cref{sec:logic}).
\item Using the technique of logical relations, we build a type system on top of \thelogic.
  By building a type system on top of separation logic, we can establish non-interference of programs that consist of typed and untyped parts (\Cref{sec:logrel}).
\item To compose proofs of program modules that cannot be type checked (because their interface relies on functional correctness), we show how to express modular separation logic specifications for non-interfence in \thelogic (\Cref{sec:modular_spec}).
\item We prove soundness of \thelogic by constructing a bisimulation out of a separation logic proof (\Cref{sec:soundness}).
\item We have mechanized \thelogic, its type system, its soundness proof, and all examples in the paper and appendix, in Coq (\Cref{sec:formalization}).
  The mechanization can be found online at~\cite{seloc_coq}.
\end{itemize}

%%% Local Variables:
%%% mode: latex
%%% TeX-master: "main"
%%% End:

%% file: examples.tex
\section{Motivating Examples}
\label{sec:examples}
Before we proceed with the formal development of the paper in \Cref{sec:preliminaries}, we present a number of challenging programs % (written in an ML-like language)
 to demonstrate the expressivity of \thelogic. % and to motivate the design choices we made.

\subsection{Modularity and data structures}
\label{sec:set_example}

To guarantee non-interference, one should prove that \emph{high-sensitivity} (\ie confidential) information cannot leak via \emph{low-sensitivity} (\ie publicly observable) outputs.
Apart from such explicit leaks, one has to prove the absence of implicit leaks that arise from the timing behavior of the program.
To avoid timing leaks, Agat and Sands \cite{AgatSands01} outlined the ``worst-case principle'': a non-interfering algorithm operating on high-sensitivity data should have the same best-case and worst-case execution time.
We apply this design principle to a set data structure that stores high-sensitivity elements.
The implementation can be type checked using our approach, automatically providing a proof of timing-sensitive non-interference.

To encapsulate the internal set representation, we first present the interface of our data structure.
This interface is given using \emph{closures} (\ie higher-order functions):\footnote{When using modules or classes, the same kind of considerations apply.}
\begin{align*}
\<val>& \ \mathit{new\_set} : \tunit \to 
\left\{
\begin{array}{@{} l @{\ } l}
\mathit{lookup} &: \tint^\high \to \tbool^\high;\\
\mathit{insert} &: \tint^\high \to \tunit
\end{array}
\right\}
\end{align*}
The function $\mathit{new\_set}$ allocates an empty set, and returns a record with functions that operate on the set.
The function $\mathit{lookup}$ takes a high-sensitivity integer---typed as $\tint^\high$, where $\high$ refers to the \emph{high}-sensitivity of the data--- and returns a high-sensitivity Boolean---typed as $\tbool^\high$---that signifies whether the argument is in the set or not.
The function $\mathit{insert}$ takes a high-sensitivity integer, and adds it to the set.

\begin{figure}[t]
\begin{align*}
\<let>& \ \mathit{new\_set}\ \unittt = \\
&\Let k = \Ref(1) in \\
&\Let arr = \Ref(\arrayMake{1}{\<None>}) in\\
&
\left\{
\begin{array}{@{} l @{}}
\mathit{lookup}\ x =
\begin{array}[t]{@{} l @{}}
  \\\hspace{-4em}
  \mathit{lookup\_loop}\ (\deref\mathit{arr})\ 
     (\deref k)\ 0\ (\mathit{cap}\ (\deref k))\ x\ \<false> % in \\
\end{array}\\
\mathit{insert}\ x =
\begin{array}[t]{@{} l @{}}
  \mathit{insert\_loop}\ \mathit{arr}\ \mathit{k}\ 0\ x %;\\
\end{array}
\end{array}
\right\} \\
\<let> & \<rec> \ \mathit{lookup\_loop}\ a\ k\ l\ r\ x\ \mathit{is\_found} = \\
& \If k = 0 then \mathit{is\_found} \Else \\
& \Let i = (l+r)/2 in \\
& \Let e = \arrayGet{a}{i} in \\
& \Let \mathit{lr_1} = (i+1, r) in
 \Let \mathit{lr_2} = (l, i-1) in \\
& \Let \mathit{(l, r)} = \If (e < x) then \mathit{lr_1} \Else \mathit{lr_2} in \\
& \mathit{lookup\_loop}\ a\ (k-1)\ l\ r\ x\ (\mathit{is\_found} \mathrel{\vee} (e = x)) \\
\<let> & \<rec>\ \mathit{insert\_loop}\ arr\ k\ i\ x = \dots
\end{align*}
\caption{Implementation of a set using the ``worst-case principle''.}
\label{fig:set1}
\end{figure}

\Cref{fig:set1} shows an implementation of our set interface using a sorted dynamic array.\footnote{The full implementation can be found in the Coq mechanization.
The full implementation moreover makes use of locks to obtain thread-safety.}
To implement the function $\mathit{lookup}$, we make use of binary search---but with a twist to avoid timing leaks.
An ordinary version of binary search would terminate once it has found the element, making it possible to observe if the element is in the set via timing.
Our implementation ensures that $\mathit{lookup}$ takes the same time regardless of whether the element is in the set. 
To achieve that, we represent the set using an array whose size $n$ satisfies $\mathit{cap}(k) = n$, for some $k$:
\begin{align*}
\mathit{cap}(0) ={}& 0 &
\mathit{cap}(k+1) ={}& 1 + 2 \cdot \mathit{cap}(k)
\end{align*}
This guarantees that the array can be recursively partitioned into two sub-arrays of the same size and a pivot element in the middle.
If the number of actual elements in the set is less than $\mathit{cap}(k)$, the array is padded with a dummy element.\footnote{For comparisons $<$ and equality $=$ checks we assume that the dummy element $\<None>$ is the greatest element and that it is not equal to any actual element in the array, which are of the form $\<Some>(x)$.}

If at some iteration of $\mathit{lookup\_loop}$ we find that the element $x$ is present in the array, we make note of that fact but still continue with the recursion until the array is no longer splittable.
Thus, the function $\mathit{lookup}$ is always executed with $k$ levels of recursion for an array of size $\mathit{cap}(k)$.
In the implementation of $\mathit{lookup\_loop}$ we pass the parameter $k$ and decrease it on every recursive call.

The function $\mathit{insert}$ traverses the whole array and is thus always executed with $\mathit{cap}(k)$ levels of recursion.
If the array is full, then it is dynamically resized to the size $\mathit{cap}(k+1)$.
In summary, both $\mathit{lookup}$ and $\mathit{insert}$ operations employ a low-sensitivity termination condition.

We use our type system (described in \Cref{sec:logrel}) to type check the implementation against the interface.
Of special note here is the type checking of the $\<if>$ branching.
In the implementation of $\mathit{lookup\_loop}$ and $\mathit{insert\_loop}$ we branch on high-sensitivity data.
Notably, in $\mathit{lookup\_loop}$ we compare the argument $x$ with the pivot $e$ (both are high-sensitivity integers), and descend into one of the partitions of the array depending on this comparison.
Branching on high-sensitivity data is not secure in general, but in this case the branching is secure.
This is 
% However, in our particular implementation this branching on high data is benign,
because both branches simply return variables ($\mathit{lr_1}$ and $\mathit{lr_2}$), \ie they do not perform any
computations, and thus do not leak information about the high-sensitivity condition via timing.\footnote{In a low-level language like C the branching can be written using arithmetic.}

\subsection{Typing via manual proof}
\label{sec:examples:manual}

The example in the previous section made use of various operations on arrays: $\mathit{array\_make}$, $\mathit{array\_get}$, and $\mathit{array\_set}$.
When reasoning about the set data structure, we assumed that these array operations are safe and secure, \ie when one tries to access an out-of-bounds index, $\mathit{array\_get}$ returns a dummy element, instead of reading arbitrary memory.

The programming language that we consider does not have safe arrays as a primitive construct.
Instead, safe arrays are implemented as a library: an array is stored together with its length, and the unsafe operations are wrapped in dynamic checks.
Naturally, such operations cannot be type checked in an ML-style type system, because their safety and security depends on functional correctness.
However, one of the core features of our approach is that such functions can be assigned types through a manual separation logic proof in \thelogic.
Such a manual proof takes functional properties (e.g., that the index is within the array bounds) into account.
Once we manually verify that the array library satisfies the desired typing, we can compose it with the type checked example from the previous section to obtain a library that guarantees safety and non-interference for its clients.\footnote{The proof of the array library and its integration in the type checking of the set data structure can be found in the Coq mechanization.}

The combination of typing and manual proof is important for compositionality and scalability: challenging library code whose security relies on functional correctness (such as the library for safe arrays) can be manually verified using separation logic, and then used to automatically type check other libraries (such as the set data structure).

\subsection{Fine-grained concurrency}
\label{sec:example:classify}
\begin{figure}[t!]
\begin{align*}
\<let> \<rec> \mathit{thread1}\;\mathit{out}\;\mathit{r} ={}&
  \begin{array}[t]{@{} l @{} l @{}}
  ( & \<if> \neg \deref\mathit{r}.\mathit{is\_classified} \\[-0.2em]
    & \<then> \mathit{out} \gets \deref\mathit{r}.\mathit{data}\ \<else> \unittt);
  \end{array}\\[-0.2em]
  & \mathit{thread1}\; \mathit{out}\; \mathit{r} \\
\<let> \mathit{thread2}\;\mathit{r} ={}&
  \mathit{r}.\mathit{data} \gets 0; \\[-0.2em]
  & \mathit{r}.\mathit{is\_classified} \gets \<false>\\
\<let> \mathit{prog}\; \mathit{out}\; \mathit{secret} ={}&
  \begin{array}[t]{@{} l @{}}
  \Let \mathit{r} = \left\{
  \begin{array}{@{} l @{}}
  \mathit{data} = \Ref(\mathit{secret}); \\[-0.2em]
  \mathit{is\_classified} = \Ref(\<true>)
  \end{array} \right\} \end{array} \\
  & in \!\apar{\mathit{thread1}\; \mathit{out}\; \mathit{r}}{\mathit{thread2}\; \mathit{r}}
\end{align*}
\caption{Lock-free value-dependent classification.}
\label{fig:value_dep}
\end{figure}
As shown in \Cref{sec:examples:manual}, the ability to fall back to a manual proof is useful to assign types to code that uses operations such as array indexing whose safety and security relies on functional correctness.
This ability becomes even more pertinent for (fine-grained) concurrent programs, where the safety and security can depend on specific protocols for accessing data that is shared between threads.

To demonstrate the application to concurrency, we consider the program $\mathit{prog}$ in \Cref{fig:value_dep}, which is a lock-free version of a similar lock-based program in \cite{SecCSL}.
The program runs two threads in parallel, both of which operate on a reference $\mathit{r}.\mathit{data}$.
The data in this reference has a \emph{value-dependent classification}: the value of the flag $\mathit{r}.\mathit{is\_classified}$ determines the sensitivity of $\mathit{r}.\mathit{data}$.
If the flag $\mathit{r}.\mathit{is\_classified}$ is set to $\<false>$, then the data stored in $\mathit{r}.\mathit{data}$ is classified with low-sensitivity, and if it is set to $\<true>$, the the data is classified with high-sensitivity.
The record $r$ initially contains high-sensitivity data from the integer variable $\mathit{secret}$.
The first thread $\mathit{thread1}$ checks if the record $r$ is classified (\ie the flag $\mathit{r}.\mathit{is\_classified}$ is $\<true>$), and if it is not, it leaks the data $\mathit{r}.\mathit{data}$ to an attacker-observable channel $\mathit{out}$.
The second thread $\mathit{thread2}$ overwrites the data stored in $r$ and resets the classification flag.

Due to the precise interplay of the two threads, the program $\mathit{prog}$ is secure, in the sense that it does not leak the data $\mathit{secret}$ onto the public channel $\mathit{out}$.
Since our example does not use locks, there are more possible interleavings than in the original example in \cite{SecCSL}, and consequently there are more things that could potentially go wrong in $\mathit{thread1}$:
\begin{enumerate}
\item the data $\mathit{r}.\mathit{data}$ can still be classified even if the bit $\mathit{r}.\mathit{is\_classified}$ is set to $\<false>$;
\item the classification of the data stored in $\mathit{r}$ might change between reading the field $\mathit{is\_classified}$ and reading the actual data from the field $\mathit{data}$.
\end{enumerate}
Notice that if we replace the second thread by the expression below, where the two operations in $\mathit{thread2}$ have been swapped, then we would violate the first condition:% (between the second and the third lines of $\mathit{thread2}_{\mathit{bad}}$).
\begin{align*}
\<let> \mathit{thread2}_{\mathit{bad}}\;\mathit{r} ={}&
  \mathit{r}.\mathit{is\_classified} \gets \<false>;\ 
  \mathit{r}.\mathit{data} \gets 0
\end{align*}

To verify that both of these situations cannot occur, we have to establish a \emph{protocol} on accessing the record $\mathit{r}$.
The protocol should ensure that at the moment of reading $\mathit{r}.\mathit{is\_classified}$ the data $\mathit{r}.\mathit{data}$ has the correct classification (ruling out situation~1).
The protocol should also ensure a form of \emph{monotonicity}: whenever the classification becomes low (\ie $\mathit{r}.\mathit{is\_classified}$ becomes $\<false>$), $\mathit{r}.\mathit{data}$ is not going to contain high-sensitivity data for the rest of the program (ruling out situation~2).

The security of $\mathit{thread1}$, and the whole program, depends on the specific protocol attached to the record $\mathit{r}$ and that the protocol is followed by all the components that operate on it.
In particular, for this example the security depends on the fact that classification only changes in a \emph{monotone} way.
We outline the proof of safety and security of this example in \Cref{sec:logic:value_dep}.

\subsection{Higher-order functions and dynamic references}
\label{sec:awkward_example}
As shown in this section, higher-order functions are useful for modularity---they can be used to model interfaces.
However, since they can operate on encapsulated state, they are difficult to reason about.
Fortunately, \thelogic's protocol mechanism is also applicable to proving non-interference of functions with encapsulated state.
Consider the program $\mathit{awk}$, a variation of the ``awkward example'' of Pitts and Stark \cite{pitts/stark:operfl}:
\[
\<let> \mathit{awk}\; \val  = \Let x = \Ref(\val) in \!\!
  \Lam f. x \gets 1; f (); \deref x
\]
When applied to a value $\val$, the program $\mathit{awk}$ returns a closure that, when invoked, always returns low-sensitivity data from the reference $x$, even if the original value $v$ has high-sensitivity.
Intuitively, $\mathit{awk}\;\val$ returns a closure that does not leak any data, even if the original value $\val$ passed to $\mathit{awk}$ had high-sensitivity.
The lack of leaks crucially relies on the following facts:
\begin{itemize}
\item the reference $x$ is allocated in, and remains local to, the closure, it cannot be accessed without invoking the closure;
\item the reference $x$ can be updated only in a monotone way: once the original value $v$ gets overwritten with $1$, the reference $x$ never holds a high-sensitivity value again.
\end{itemize}
To see why second condition is important, consider $\mathit{awk}_{\mathit{bad}}$, which violates the monotonicity, and is thus not secure:
\begin{align*}
\<let> \mathit{awk}_{\mathit{bad}}\; \val ={}& \Let x = \Ref(v) in \!\!
   \Lam f. x \gets v; x \gets 1; f (); \deref x
\end{align*}
Let $h = \mathit{awk}_{\mathit{bad}}\;v$ for a high-sensitivity value $\val$.
Now, when running $h\;(\Lam x. \Fork{h(\mathit{id})})$, an attacker could influence the scheduler so that the first dereference $\deref x$ happens just after the assignment $x \gets v$ in the forked-off thread, causing $\val$ to leak.

Pitts and Stark studied the ``awkward example'' to motivate the difficulties of reasoning about higher-order functions and state.
They were interested in contextual equivalence, but as we can see, similar considerations apply to non-interference.

%%% Local Variables:
%%% mode: latex
%%% TeX-master: "main"
%%% End:

%% file: preliminaries.tex
\section{Preliminaries}
\label{sec:preliminaries}

In this section we describe the programming language that we consider in this paper (\Cref{sec:lang}), and the non-interference property that \thelogic establishes (\Cref{sec:security}).

\subsection{Object language and scheduler semantics}
\label{sec:lang}
\thelogic is defined over an ML-like programming language, called $\HeapLang$,
 with higher-order mutable references, recursion, the $\<fork>$ operation, and atomic compare-and-swap $\<CAS>$.
$\HeapLang$ is the default programming language that is shipped with Iris~\cite{irisWWW}.
Its values and expressions are:
\begin{align*}
\val \in \Val \bnfdef{}&
	\Rec f \var = \expr \ALT
        (\val_1, \val_2) \ALT
        \<true> \ALT \<false> \ALT
        \dots\\
\expr, \exprB, t \in \Expr \bnfdef{}&
	\var \ALT
	\Rec f \var = \expr \ALT
	\expr_1 (\expr_2) \ALT
	\Fork{\expr} \\ & \hspace{-3em} \ALT
	\Ref(\expr) \ALT
	\deref \expr \ALT
	\expr_1 \gets \expr_2 \ALT
        \<CAS>(\expr_1,\expr_2,\expr_3) \ALT \dots
\end{align*}

We omit the usual operations on pairs, sums, and integers.
We use the following syntactic sugar:
$(\Lam \var. \expr) \eqdef (\Rec \_ \var = \expr)$,
$(\Let \var = \expr_1 in \expr_2) \eqdef ((\Lam \var. \expr_2)\ \expr_1)$, and
$(\expr_1 ; \expr_2) \eqdef (\Let \_ = \expr_1 in \expr_2)$.
$\HeapLang$ has no primitive syntax for records, so they are modeled using pairs.
Arrays are omitted in the paper, but they are present in the Coq mechanization.

$\HeapLang$ features dynamic thread creation, so we can implement the parallel composition operation using $\<fork>$:
\begin{align*}
\<let>\Rec \mathit{join} x ={} & \Match \deref{x} with \<Some>(v) \to v\\
  & \mid \<None> \to \mathit{join}\ x\\
\<let>\mathit{par}(f_1, f_2) ={} & \Let x = \Ref(\<None>) in \\
& \Fork{x \gets \<Some> (f_1 \unittt)} \\
& \Let v_2 = f_2 \unittt in (\mathit{join}\ x, v_2)\\
\apar{\expr_1}{\expr_2} \eqdef{}& \mathit{par}(\Lam \_. \expr_1, \Lam \_. \expr_2)
\end{align*}

The operational semantics of $\HeapLang$ is split into three parts:
thread-local head reduction \(\hstep\), thread-local reduction \(\step\), and thread-pool reduction \(\pstep\).
The thread-local head reduction is of the form $(\expr_1, \stateS_1) \hstep (\expr_2, \stateS_2)$, 
where $\expr_i$ is an expression, and $\stateS_i$ is a heap, \ie a finite map from locations to values ($\State \eqdef \Loc \fpfn \Val$).
Since the security condition that we consider (\Cref{sec:security}) is tailored towards a deterministic thread-local semantics, we parameterize the operational semantics by an \emph{allocation oracle} $A : \State \to \Loc$: a function from heaps to locations satisfying $A(\stateS) \not\in \stateS$.
With the allocation oracle, the allocation head reduction is as follows:
\[
\infer{}
{(\Ref(\val), \stateS) \hstep (A(\stateS), \mapinsert{A(\stateS)}{\val}{\stateS})}
\]
The other rules for the head reduction relation are standard and can be found in the Coq mechanization.

The thread-local head reduction is lifted to the thread-local reduction using \emph{call-by-value evaluation contexts}:
\[
\lctx \in \Ectx \bnfdef{}
  \exphole \ALT
  \lctx(\val_2) \ALT
  \expr_1(\lctx) \ALT
  \If \lctx then \expr_1 \Else \expr_2 \ALT \dots
\]
The thread-local reduction is of the form $(\expr_1, \stateS_1) \step (\vec{\expr}_2, \stateS_2)$.
The second component contains a list $\vec{\expr}_2$ of expressions to accommodate forked-off threads as in \ruleref{step-fork}:
\begin{mathpar}
\inferH{step-lift}
{(\expr_1, \stateS_1) \hstep (\expr_2, \stateS_2)}
{(\fill\lctx[\expr_1], \stateS_1) \step (\fill\lctx[\expr_2], \stateS_2)}
\quad
\inferH{step-fork}
{\vec\expr = \fill\lctx[\unittt]\ \expr}
{(\fill\lctx[\Fork{\expr}], \stateS) \step (\vec\expr, \stateS)}
\end{mathpar}

The thread-pool reduction \(\pstep\) is defined by lifting the thread-local reduction to \emph{configurations} $(\vec{\expr}, \stateS)$.
Here, $\vec{\expr}$ contains all threads, including values for the threads that have terminated.
In the definition of \(\pstep\) we non-deterministically select an expression to take a thread-local step:
\[
\infer
{(\expr_i, \stateS_1) \step (\expr'_i \vec\expr, \stateS_2)}
{(\expr_0 \dots \expr_i \dots \expr_n, \stateS_1)
\pstep
(\expr_0 \dots \expr'_i \dots \expr_n\vec\expr, \stateS_2)
}
\]

\subsection{Strong low-bisimulations}
\label{sec:security}
To state the soundness theorem of \thelogic in \Cref{sec:soundness_statement}, we adapt a timing-sensitive notion of non-interference for concurrent programs known as \emph{strong low-simulations} on configurations by Sabelfeld and Sands~\cite{prob-ni}.
To define this notion, we first fix a set $\Llocset \subseteq \Loc$ of \emph{output locations}, which we assume to be low-sensitivity observable locations.
For simplicity, we require these locations to contain integers.

\begin{definition}
Heaps $\stateS_1$ and $\stateS_2$ are \emph{low-equivalent for output locations $\Llocset \subseteq \Loc$}, notation $\stateS_1 \lowEquiv \stateS_2$, if they agree on all the $\Llocset$-locations, \ie
$
  \All \loc \in \Llocset. \stateS_1(\loc) = \stateS_2(\loc) \neq \bot \wedge \stateS_1(\loc) \in \mathbb{Z}
$.
\end{definition}

\begin{definition}
\label{def:strong-low-bisim}
A \emph{a strong low-bisimulation} is a partial equivalence (\ie symmetric and transitive) relation $\RR$ on configurations such that:
\begin{enumerate}
\item If $(\val\vec{\expr}, \stateS_1) \RR (\valB\vec{\exprB}, \stateS_2)$, then $\val = \valB$;
\item If $(\vec{\expr}, \stateS_1) \RR (\vec{\exprB}, \stateS_2)$, then $\lvert \vec{\expr} \rvert = \lvert \vec{\exprB}\rvert$ and $\stateS_1 \lowEquiv \stateS_2$;
\item If $(\expr_0 \dots \expr_i \dots \expr_n, \stateS_1) \RR (\exprB_0 \dots \exprB_i \dots \exprB_n, \stateS_2)$ and $(\expr_i, \stateS_1) \step (\expr'_i \vec{\expr}, \stateS'_1)$, then there exist an $\exprB'_i$, $\vec{\exprB}$ and $\stateS'_2$ such that:
  \begin{itemize}
  \item $(\exprB_i, \stateS_2) \step (\exprB'_i \vec{\exprB}, \stateS'_2)$;
  \item $(\expr_0 \dots \expr'_i \dots \expr_n \vec{\expr}, \stateS'_1) \RR
    (\exprB_0 \dots \exprB'_i \dots \exprB_n \vec{\exprB}, \stateS'_2)$.
  \end{itemize}
\end{enumerate}
\end{definition}

Notice that the first expression in the thread-pool is the main thread.
The first condition in \Cref{def:strong-low-bisim} thus states that the return values of the main-thread should agree.

To model the input/high-sensitivity data we use free variables.
For simplicity we assume that the input data consists of integers.
We then arrive at the following top-level definition of security.

\begin{definition}[Security]
\label{def:security}
An expression $\expr$ with free variables $\vec{x}$ is \emph{secure} if for any heap $\stateS$ with $\stateS \lowEquiv \stateS$, and any sequences of integers $\vec{i}$, $\vec{j}$ with $\lvert \vec{i} \rvert = \lvert \vec{j} \rvert = \lvert \vec{x} \rvert$, there exists a strong low-bisimulation $\RR$ such that
$
  (\subst{\expr}{\vec{x}}{\vec{i}}, \stateS) \RR (\subst{\expr}{\vec{x}}{\vec{j}}, \stateS).
$
\end{definition}

\subsection{Non-determinism and non-interference}
\label{sec:rand_example}
The semantics presented in \Cref{sec:lang} is deterministic on the thread-local level, but we can still account for non-determinism arising from a scheduler.
Consider the program $\<rand>$, which uses intrinsic non-determinism of the thread-pool semantics to return either $\<true>$ or $\<false>$:
\begin{equation*}
\begin{aligned}
\<let> \<rand>\; () ={} & \Let x = \Ref(\<true>) in \!\Fork{x \gets \<false>}; \deref x
\end{aligned}
\end{equation*}
This program is secure \wrt \Cref{def:security} (we will prove this in \Cref{sec:logic} using \thelogic).

It is worth pointing out that if we modify the program and insert an additional assignment of a high-sensitivity value $h$ to $x$, then the resulting program is \emph{not} secure:
\begin{equation*}
\begin{aligned}
\<let> \<rand>_{\mathit{bad}}\; () ={} & \Let x = \Ref(\<true>) in \\
& \Fork{x \gets h}; \Fork{x \gets \<false>}; \deref x
\end{aligned}
\end{equation*}
The program is not secure because an attacker can pick a scheduler that always executes the leaking assignment, or, even simpler, can run the program many times under the uniform scheduler.
Because the program is not secure, we cannot prove it in \thelogic.
In \thelogic, we would verify each thread separately, and we would not be able to verify the forked-off thread $x \gets h$ (precisely because it makes the non-determinism of assignments to the reference $x$ dangerous).

%%% Local Variables:
%%% mode: latex
%%% TeX-master: "main"
%%% End:

%% file: logic.tex
\section{Overview of \thelogic}
\label{sec:logic}
We provide an overview of \thelogic by presenting its proof rules for relational reasoning (\Cref{sec:dwp}), its invariant mechanism (\Cref{sec:logic:example}), its soundness theorem (\Cref{sec:soundness_statement}), and finally its protocol mechanism (\Cref{sec:logic:value_dep}), which we apply to the verification of the program $\mathit{prog}$ from \Cref{sec:example:classify}.
The grammar of \thelogic is:
\begin{align*}
\prop, \propB \in{}& \Prop \bnfdef{}
	\TRUE \ALT \FALSE \ALT
	\All \var.\prop \ALT
	\Exists \var.\prop \ALT
	\prop * \propB \\ & \ALT
	\prop \wand \propB \ALT
	\loc \mapstoLR \val \ALT
        \awp[\LRvar]{\expr}{\pred}  \tag{$\LRvar \in \{\mathsf{L}, \mathsf{R}\}$} \\
                   & \ALT
        \dwpre{\expr_1}{\expr_2}[\mask]{\pred}
                  \\ & \ALT
	\knowInv\namesp\prop \ALT
	\later \prop \ALT
	\always \prop \ALT
	\pvs[\mask_1][\mask_2] \prop \ALT \dots
\end{align*}
\thelogic features the standard separation logic connectives like separating conjunction ($\ast$) and magic wand ($\wand$).
Since \thelogic is based on Iris~\cite{iris1,iris2,iris3,irisJFP}, it incorporates all the Iris connectives and modalities, in particular the \emph{later modality} ($\later$) for dealing with recursion, the \emph{persistence modality} ($\always$) for dealing with shareable resources, and the \emph{invariant connective} ($\knowInv\namesp\prop$) and the \emph{update modality} ($\pvs[\mask_1][\mask_2]$) for establishing and relying on protocols.
We will not introduce the Iris connectives in detail, but rather explain them on a by-need basis.
An interested reader is referred to \cite{irisJFP,irisLN} for further details.
Various connectives are annotated with \emph{name spaces} $\namesp \in \InvName$ and \emph{invariant masks} $\mask \subseteq \InvName$ to handle some bookkeeping.
When the mask is omitted, it is assumed to be $\top$, the largest mask.
We let $\pvs[\mask]$ denote $\pvs[\mask][\mask]$.
Readers who are unfamiliar with Iris can safely ignore the name spaces and invariant masks.

A selection of proof rules of \thelogic is given in \Cref{fig:dwp_rules}.
Each inference rule $\infer{\prop_1 \hspace{2pt} \dotsc \hspace{2pt} \prop_n}{\propB}$ in this paper should be read as an entailment $\prop_1 \ast \dotsc \ast \prop_n \vdash \propB$.
In the subsequent sections we explain and motivate the rules of \thelogic.

\subsection{Relational reasoning}
\label{sec:dwp}
The quintessential connective of \thelogic is the \emph{double weakest precondition} $\dwpre{\expr_1}{\expr_2}[\mask]{\pred}$.
It intuitively expresses that any two runs of $\expr_1$ and $\expr_2$ are related in a lock-step bisimulation-like way, and that the resulting values of any two terminating runs are related by the \emph{postcondition} $\pred : \Val \to \Val \to \Prop$.
We refer to $\expr_1$ (resp.~$\expr_2$) as the \emph{left-hand side} (\resp the \emph{right-hand side}).
The double weakest precondition is defined such that if $\All \vec i\,\vec j \in \mathbb{Z}. \dwpre{\subst{\expr}{\vec x}{\vec i}}{\subst{\expr}{\vec x}{\vec j}}{\Ret \val_1\, \val_2. \val_1 = \val_2}$ (with $\vec x$ the free variables of $e$), then $\expr$ is secure.
We defer the precise soundness statement to \Cref{sec:soundness_statement}.

%If an expression is atomic (denoted as $\atomic(e)$), meaning that it reduces to a value in a single step, then we can use \ruleref{dwp-atomic}, in combination with an invariant opening rule, to obtain the resources from the invariant for the duration of one step of the execution.\robbert{The invariant opening rule changed. Also make clear that this rule is adapted from Iris, but generalized to the binary case.}
%We will see how to use this rule in \Cref{sec:logic:example}.

\begin{figure*}
  \begin{mathpar}
\inferH{dwp-val}
{\pred(\val_1,\val_2)}
{\dwpre{\val_1}{\val_2}[\mask]{\pred}}
\and

\inferH{dwp-wand}
{\dwpre{\expr_1}{\expr_2}[\mask]{\predB}
\and (\All \val_1\,\val_2. \predB(\val_1, \val_2) \wand \pred(\val_1, \val_2))}
{\dwpre{\expr_1}{\expr_2}[\mask]{\pred}}
\and

\inferH{dwp-fupd}
{\pvs[\mask]\dwpre{\expr_1}{\expr_2}[\mask]{\Ret \val_1\,\val_2. \pvs[\mask]\pred(\val_1, \val_2)}}
{\dwpre{\expr_1}{\expr_2}[\mask]{\pred}}
\and

\inferH{dwp-bind}
{\dwpre{\expr_1}{\expr_2}{\Ret \val_1\, \val_2.%
  \dwpre{\fill\lctx_1[\val_1]}{\fill\lctx_2[\val_2]}{\pred}}}
{\dwpre{\fill\lctx_1[\expr_1]}{\fill\lctx_2[\expr_2]}{\pred}}
\and

\inferH{dwp-pure}
{\pureexec{\expr_1}{\expr'_1}
\and \pureexec{\expr_2}{\expr'_2}
\and \later \dwpre{\expr'_1}{\expr'_2}{\pred}
}
{\dwpre{\expr_1}{\expr_2}{\pred}}
\and

\inferH{dwp-fork}
{\later \dwpre{\expr_1}{\expr_2}{\TRUE} \and
  \later \pred (\unittt, \unittt)}
{\dwpre{(\Fork{\expr_1})}{(\Fork{\expr_2})}[\mask]{\pred}}
\and

\inferhref{dwp-awp}{dwp-atomic-wp}
{\awp[\mathsf{L}]{\expr_1}{\predB_1} \and
  \awp[\mathsf{R}]{\expr_2}{\predB_2} \and
  (\All \val_1, \val_2. (\predB_1(\val_1) \ast \predB_2(\val_2)) \wand \later\pred(\val_1, \val_2))}
{\dwpre{\expr_1}{\expr_2}[\mask]{\pred}}
\and
\inferH{awp-store}
 {\loc \mapstoLR \val_1 \and (\loc \mapstoLR \val_2 \wand \pred\unittt)}
{\awp[\LRvar]{\loc \gets \val_2}{\pred}}
\and
\inferH{awp-load}
{\loc \mapstoLR \val \and (\loc \mapstoLR \val \wand \pred(\val))}
{\awp[\LRvar]{\deref\loc}{\pred}}
\and
\inferH{awp-alloc}
{\All \loc. \loc \mapstoLR \val \wand \pred(\loc)}
{\awp[\LRvar]{\Ref(\val)}{\pred}}
\and
\inferH{dwp-inv-alloc}
{\prop
\and (\knowInv{\namesp}{\prop} \wand \dwpre{\expr_1}{\expr_2}{\pred})}
{\dwpre{\expr_1}{\expr_2}{\pred}}
\and
\inferH{inv-dup}{\knowInv{\namesp}{\prop}}{\knowInv{\namesp}{\prop} * \knowInv{\namesp}{\prop}}
\and
\inferhref{dwp-inv}{dwp-atomic-inv}
{\knowInv{\namesp}{\prop} \and
(\later P \wand
  \dwpre{\expr_1}{\expr_2}[\mask - \namesp]{\Ret \val_1\, \val_2. \prop \ast \pred(\val_1, \val_2)}) \and
\atomic(\expr_1) \and \atomic(\expr_2) \and
\namesp \in \mask}
{\dwpre{\expr_1}{\expr_2}[\mask]{\pred}}
\end{mathpar}
\caption{A selection of the proof rules of \thelogic.}
\label{fig:dwp_rules}
\end{figure*}

A selection of rules for double weakest preconditions\footnote{Some of the \thelogic rules involve the \emph{later modality} $\later$, which is standard for dealing with recursion and impredicative invariants \cite[Section 5.5]{irisJFP}.
The occurrences of $\later$ can be ignored for the purposes of this paper.} are given in \Cref{fig:dwp_rules}.
Some of these rules are generalizations of the ordinary weakest precondition rules (\eg \ruleref{dwp-val}, \ruleref{dwp-wand}, \ruleref{dwp-fupd}, \ruleref{dwp-bind}).
The more interesting rules are the \emph{symbolic execution} rules, which allow executing the programs on both sides in a lock-step fashion.
If both sides involve a pure-redex, we can use \ruleref{dwp-pure}.
The premises $\pureexec{\expr}{\expr'}$ denote that $\expr$ deterministically reduces to $\expr'$ without any side-effects (\eg $\pureexec{(\If \<true> then \expr \Else t)}{\expr}$).
If both sides involve a fork, we can use the rule \ruleref{dwp-fork}, which is a generalization of Iris's fork rule to the relational case.
To explain \thelogic's rules for symbolic execution of heap-manipulating expressions, we need to introduce some additional machinery:

\begin{itemize}
\item Due to \thelogic's relational nature, there are left- and right-hand side versions of the \emph{points-to connectives} \mbox{$\loc \mapstoLR \val$}, where $\LRvar \in \{\mathsf{L}, \mathsf{R}\}$, which denote that the value $\val$ of location $\loc$ in the heap associated with the left-hand side program and the right-hand side program, \resp
\item To avoid a quadratic explosion in combinations of all possible heap-manipulating expressions on the left- and the right-hand side, \thelogic includes a unary weakest precondition $\awp[\LRvar]{\expr}{\pred}$ for atomic and fork-free expressions.
  The rules for unary weakest preconditions (\eg \ruleref{awp-store}, \ruleref{awp-load}, \ruleref{awp-alloc}) are similar to those of Iris, but each rule is parameterized by a side $\LRvar \in \{\mathsf{L}, \mathsf{R}\}$.
\end{itemize}

The rule \ruleref{dwp-atomic-wp} connects $\dwp$ and $\awpsymb$.
For instance, using \ruleref{dwp-atomic-wp}, \ruleref{awp-store}, and \ruleref{awp-load}, we can derive:
%the following symbolic execution rule:
\begin{mathpar}
\infer
{\loc_1 \mapstoL \val_1 \and
 \loc_2 \mapstoR \val_2 \and \\
  (\loc_1 \mapstoL \val_1 \ast \loc_2 \mapstoR \val'_2) \wand \dwpre{\val_1}{\unittt}{\pred}
}
{\dwpre{\deref \loc_1}{(\loc_2 \gets \val'_2)}{\pred}}
\end{mathpar}

\subsection{Invariants}
\label{sec:logic:example}
Let us demonstrate, by means of an example, how to use the symbolic execution rules together with the powerful invariant mechanism of Iris.
Recall the $\<rand>$ example from \Cref{sec:rand_example}.
We can use invariants to prove the following:
\begin{proposition}
\label{prop:dwp_rand}
$\dwpre{\<rand>\,()}{\<rand>\,()}{\Ret \val_1\, \val_2. \val_1 = \val_2}$.
\end{proposition}
\begin{proof}
First we use \ruleref{dwp-pure} to symbolically execute a $\beta$-reduction.
We then use \ruleref{dwp-bind} to ``focus'' on the $\Ref(\<true>)$ subexpression, leaving us with the goal:
\begin{gather*}
 \dwpre{\Ref(\<true>)}{\Ref(\<true>)}{\pred} \\
\begin{aligned}
\mbox{where }
 \pred(\loc_1,\loc_2) \eqdef \dwpre{&\Let x = \loc_1 in \dots}{\\[-0.2em] & \Let x = \loc_2 in \dots}{\Ret \val_1\, \val_2. \val_1 = \val_2}
\end{aligned}
\end{gather*}
We then symbolically execute the allocation, using \ruleref{dwp-atomic-wp} and \ruleref{awp-alloc}, obtaining
$\loc_1 \mapstoL \<true>$ and $\loc_2 \mapstoR \<true>$:
\begin{equation*}
\begin{aligned}
  \loc_1 \mapstoL \<true> & \ast \loc_2 \mapstoR \<true> \\
  \proves
  \dwpre{& \Fork{\loc_1 \gets \<false>}; \deref \loc_1}
     {\\[-0.2em] & \Fork{\loc_2 \gets \<false>}; \deref \loc_2}{\Ret \val_1\, \val_2. \val_1 = \val_2}
\end{aligned}
\end{equation*}

It is tempting to use \ruleref{dwp-fork}; but in both the main thread and the forked-off thread we need $\loc_1 \mapstoL -$ and $\loc_2 \mapstoR -$ to symbolically execute the dereference and assignment to $\loc_1$ and $\loc_2$.
To share the points-to connectives between both threads, we put them into an Iris-style invariant.

Iris-style invariants are logical propositions denoted as $\knowInv{\namesp}{\prop}$, which express that $\prop$ holds at all times.
Unlike in other logics, Iris-style invariants are not attached to locks.
Rather, one can explicitly open an invariant during an atomic step of execution to get access to its contents.
To create a new invariant we use the \ruleref{dwp-inv-alloc} rule,
which transfers $\prop$ into the an invariant $\knowInv{\namesp}{\prop}$ with a name space $\namesp \in \InvName$.
The transfer of $\prop$ into an invariant makes it possible to share $\prop$ between different threads (using \ruleref{inv-dup}).
To access an invariant we use the rule \ruleref{dwp-atomic-inv}.
It allow us to \emph{open} an invariant during an atomic symbolic execution step.
The \emph{masks} $\mask \subseteq \InvName$ on $\dwp$ are used to keep track of which invariants have been open.
This is done to prevent invariant reentrancy.

Returning to our example, we can use \ruleref{dwp-inv-alloc} to allocate the invariant
$
I \eqdef \knowInv{\namesp}{\Exists b \in \mathbb{B}. \loc_1 \mapstoL b \ast \loc_2 \mapstoR b}
$.
This invariant not only allows different threads to access $\loc_1$ and $\loc_2$ (via \ruleref{inv-dup}), but it also ensures that $\loc_1$ and $\loc_2$ contain the same Boolean value throughout the execution.

The proof then proceeds as follows.
We apply \ruleref{dwp-fork} and get two new goals:
\begin{enumerate}
\item $I \proves \dwpre{\loc_1 \gets \<false>}{\loc_2 \gets \<false>}{\TRUE}$;
\item $I \proves \dwpre{\deref \loc_1}{\deref \loc_2}{\Ret \val_1\, \val_2. \val_1 = \val_2}$.
\end{enumerate}
The invariant $I$ can be used for proving both goals (\ruleref{inv-dup}).
The first goal involves proving that the assignment of $\<false>$ to $\loc_1$ and $\loc_2$ is secure.
We verify this via \ruleref{dwp-atomic-inv}, and temporarily opening the invariant $I$ to obtain $\loc_1 \mapstoL b$ and $\loc_2 \mapstoR b$.
We then apply \ruleref{dwp-atomic-wp}, and symbolically execute the assignment to obtain $\loc_1 \mapstoL \<false>$ and $\loc_2 \mapstoR \<false>$.
At the end of this atomic step, we verify that the invariant $I$ still holds.

The second goal is solved in a similar way.
When we dereference $\loc_1$ and $\loc_2$ we know that they contain the same value because of the invariant $I$.
\end{proof}

\subsection{Soundness}
\label{sec:soundness_statement}
We now state \thelogic's soundness theorem, which guarantees that verified programs are actually secure \wrt \Cref{def:security}.

As we have described in \Cref{sec:security}, we fix a set $\Llocset$ of output locations that we assume to be observable by the attacker.
We require these locations to always contain the same data in both runs of the program.
To reflect this in the logic, we use an invariant that owns the observable locations and forces them to contain the same values in both heaps:
\begin{equation*}
  I_{\Llocset} \eqdef \Sep_{\loc \in \Llocset}\ \knowInv{\namesp.{(\loc,\loc)}}{\Exists i \in \mathbb{Z}. \loc \mapstoL i \ast \loc \mapstoR i}
\end{equation*}
When we verify a program under the invariant $I_{\Llocset}$, we are forced to interact with the locations in $\Llocset$ as if they are permanently publicly observable.
With this in mind we state the soundness theorem, which we prove in \Cref{sec:soundness}.
\begin{theorem}[Soundness]
\label{thm:soundness1}
Suppose that:
\[
I_{\Llocset} \vdash \All \vec i,\vec j \in \mathbb{Z}.\dwpre{\subst{\expr}{\vec{x}}{\vec{i}}}{\subst{\expr}{\vec{x}}{\vec{j}}}{\Ret \val_1\, \val_2. \val_1 = \val_2}
\]
is derivable, where $\vec{x}$ are the free variables of $\expr$, and $\vec{i}$ and $\vec{j}$ are lists of integers with $\lvert \vec{i} \rvert = \lvert \vec{j} \rvert = \lvert \vec{x} \rvert$, then:
\begin{itemize}
\item the expression $\expr$ is secure, and,
\item the configuration $(\subst{\expr}{\vec{x}}{\vec{i}}, \stateS)$ is safe (\ie cannot get stuck) for any list of integers $\vec i$, and any heap $\stateS$ with $\stateS \lowEquiv \stateS$.
\end{itemize}
\end{theorem}

\subsection{Protocols}
\label{sec:logic:value_dep}
Now that we have seen the basics of Iris-style invariants in \thelogic, let us use the protocol mechanism \thelogic inherits from Iris to verify the example $\mathit{prog}$ from \Cref{fig:value_dep}.
We prove the following proposition, which serves as a premise for \Cref{thm:soundness1}, and therefore implies the security of $\mathit{prog}$.

\begin{proposition}
\label{prop:dwp_prog}
For any integers $i_1, i_2 \in \mathbb{Z}$, we have
$I_{\{ \outloc \}} \proves \dwpre{\mathit{prog}\, \outloc\, i_1}{\mathit{prog}\, \outloc\, i_2}{\Ret \val_1\, \val_2. \val_1 = \val_2 = (\unittt, \unittt)}$.
\end{proposition}

\begin{proof}
We first need a derived rule for parallel composition (which we defined in terms of $\<fork>$ in \Cref{sec:lang}).
The parallel composition operation satisfies a binary version of the standard specification in Concurrent Separation Logic~\cite{OHEARN2007}:
\begin{mathpar}
\inferH{dwp-par}
{\dwpre{\expr_1}{\exprB_1}{\predB_1} \and
\dwpre{\expr_2}{\exprB_2}{\predB_2} \\
\big(\All \val_1,\val_2,\valB_1,\valB_2.
  {\begin{array}[t]{@{} l @{}}
  (\predB_1(\val_1, \valB_1) \ast \predB_2(\val_2, \valB_2)) \wand{} \\[-0.3em]
  \pred ((\val_1, \valB_1), (\val_2, \valB_2))\big)
  \end{array}}
}
{\dwpre{(\apar{\expr_1}{\expr_2})}{(\apar{\exprB_1}{\exprB_2})}{\pred}}
\end{mathpar}
Second, we need to establish a protocol on the way the values in the record $r$ may evolve.
%Following the state change in $\mathit{thread2}$,
We identify three logical states  $\STSState \eqdef \{ \Classified, \Intermediate, \Declassified \}$ the record $r$ can be in; visualized in \Cref{fig:progproof}-a:
\begin{enumerate}
\item $\Classified$, if the data stored in the record is classified, and $\mathit{r}.\mathit{is\_classified}$ points to $\<true>$;
\item $\Intermediate$, when the data stored in the record is not classified anymore, but $\mathit{r}.\mathit{is\_classified}$ still points to $\<true>$;
\item $\Declassified$, when the data stored in the record is not classified and $\mathit{r}.\mathit{is\_classified}$ points to $\<false>$.
  This state is final in the sense that once the state of the record becomes $\Declassified$, it forever
  remains so.
\end{enumerate}
The idea behind the proof is as follows: we use an invariant to track the logical state together with the points-to connectives for the physical state of the record.
This way, we ensure that the protocol is followed by both threads.

\begin{figure}
\paragraph{The protocol as a transition system}

\begin{center}
\medskip
\begin{tikzpicture}
\node[state]
    (classified)   {$\Classified$};
\node[state,right of=classified]
    (intermediate) {$\Intermediate$};
\node[state,right of=intermediate,accepting]
    (declassified) {$\Declassified$};
\draw  (classified)    edge  (intermediate)
       (intermediate)  edge             (declassified);
\end{tikzpicture}
\end{center}
\paragraph{The rules for ghost state}
\begin{center}
\small
\vspace{-1em}
\begin{mathpar}
\axiomH{state-alloc}
{\pvs[\mask] \Exists \gname. \inState{\Classified} \ast \stateToken{\Classified}}
\and
\inferH{state-agree}
{\inState{s_1} \and \stateToken{s_2}}
{s_1 = s_2}
\and
\inferH{state-change}
{\tikz[anchor=base, baseline] {\node (s1) {$s_1$};\node[right=1em of s1] (s2) {$s_2$};
  \draw (s1) edge (s2);} \and \inState{s_1} \and \stateToken{s_1}}
{\pvs[\mask] \inState{s_2}\ast\stateToken{s_2}}
\and
\inferH{declassified-dup}
{\stateToken{\Declassified}}
{\stateToken{\Declassified} \ast \stateToken{\Declassified}}
\end{mathpar}
\end{center}

\paragraph{The invariant}
\begin{equation*}
\knowInv{\namesp}{\small
\begin{aligned}
\big(&\inState{\Classified} \ast{} \Exists i_1, i_2.
   r_1.\mathit{is\_classified} \mapstoL \<true> \ast{} \\
   &r_2.\mathit{is\_classified} \mapstoR \<true> \ast{}
   r_1.\mathit{data} \mapstoL i_1 \ast{}
   r_2.\mathit{data} \mapstoR i_2 \big){} \\[0.3em]
\vee\big(&\inState{\Intermediate} \ast{} \Exists i.
   r_1.\mathit{is\_classified} \mapstoL \<true> \ast{} \\
   &r_2.\mathit{is\_classified} \mapstoR \<true> \ast{}
    r_1.\mathit{data} \mapstoL i \ast{}
   r_2.\mathit{data} \mapstoR i\big){} \\[0.3em]
\vee\big(&\inState{\Declassified} \ast{} \Exists i.
   r_1.\mathit{is\_classified} \mapstoL \<false> \ast{}\\
   &r_2.\mathit{is\_classified} \mapstoR \<false> \ast{}
    r_1.\mathit{data} \mapstoL i \ast{}
   r_2.\mathit{data} \mapstoR i \big){}
\end{aligned}}
\end{equation*}
\caption{Value-dependent classification.}
\label{fig:progproof}
\end{figure}

To model the protocol in \thelogic, we use Iris's mechanism for user-defined ghost state.
The exact way this mechanism works is described in~\cite{iris1,irisJFP}, but is not important for this paper.
What is important is that via this mechanism we can define predicates $\inState{}, \stateToken{} : \STSState \to \Prop$ that satisfy the rules in \Cref{fig:progproof}-b.
The predicate $\inState{}$ will be shared using an invariant, while $\mathit{thread2}$ will own the predicate $\stateToken{}$.
Rule \ruleref{state-agree} states that the predicates $\inState{}$ and $\stateToken{}$ agree on the logical state.
If a thread owns both predicates, it can change the logical state using \ruleref{state-change}, but only in a way that respects the
transition system.
Rule \ruleref{declassified-dup} states that once a thread learns that the record is in the final state, \ie $\Declassified$, this knowledge remains true forever.
The predicates are indexed by a \emph{ghost name} $\gname$ to allow for
different instances of the transition system using \ruleref{state-alloc}.
%This allows one to use essentially use multiple copies of the transition system in \Cref{fig:progproof}.
\Cref{fig:progproof}-c displays the invariant that ties together the ghost and physical state.
It is defined for the records $r_1$ and $r_2$ on the left- and the right-hand side, \resp
We verify each thread separately with respect to this invariant, which we open every time we access the record.

\paragraph*{Proof of the complete program}
We symbolically execute the allocation of the records $r_1$ and $r_2$, giving us the resources $r_1.\mathit{is\_classified} \mapstoL \<true>$, $r_2.\mathit{is\_classified} \mapstoR \<true>$, 
$r_1.\mathit{data} \mapstoL i_1$, and $r_2.\mathit{data} \mapstoR i_2$.
We then use \ruleref{state-alloc} to obtain $\inState{\Classified}$ and $\stateToken{\Classified}$.
With these resources at hand, we use \ruleref{dwp-inv-alloc} to establish the invariant in \Cref{fig:progproof}-c, which can be shared between both threads.
We use \ruleref{dwp-par} with $\predB_1(\val_1, \val_2) \eqdef \predB_2(\val_1, \val_2) \eqdef (\val_1 = \val_2 = \unittt)$, and use the token $\stateToken{\Classified}$ for the proof of the second thread.

\paragraph*{Proof of $\mathit{thread1}$}
We use the symbolic execution rules for dereferencing $r_1.\mathit{is\_classified}$ and $r_2.\mathit{is\_classified}$ until both of them become $\<false>$.
At that point, the invariant tells us that we are in the $\Declassified$ state.
Subsequently, when using the symbolic execution rule for dereferencing $r_1.\mathit{data}$ and $r_2.\mathit{data}$, we use a copy of the predicate $\stateToken{\Declassified}$ to determine that the last disjunct of the invariant must hold.
From that, we know that both $r_1.\mathit{data}$ and $r_2.\mathit{data}$ contain the same value.
Using this information we can safely symbolically execute the assignments to the output location $\outloc$.

\paragraph*{Proof of $\mathit{thread2}$}
We start the proof with the initial predicate $\stateToken{\Classified}$ and update the logical state with each assignment.
The complete formalized proof can be found in the Coq mechanization.
\end{proof}

%%% Local Variables:
%%% mode: latex
%%% TeX-master: "main"
%%% End:

%% file: types.tex
\section{Type system and logical relations}
\label{sec:logrel}

We show how to define a type system for non-interference as an abstraction on top of \thelogic using the technique of \emph{logical relations}.
While logical relations have been used to model type systems and logics for safety and contextual refinement in (variants of) Iris before~\cite{irisIPM,irisEffects,irisRunSt,reloc,RustBelt}, we---for the first time---use logical relations in Iris to  model a type system for non-interference (\Cref{sec:logrel:model}).
We moreover show how we can combine type-checked code with code that has been manually verified using double weakest preconditions in \thelogic (\Cref{sec:logrel:combine}).

The types that we consider are as follows:
\[
\type \in \Type \bnfdef \tunit \ALT
  \tint^\lvlvar \ALT
  \tbool^\lvlvar \ALT
  \type \times \type' \ALT
  \tref{\type} \ALT
  (\type \to \type')^\lvlvar
\]
Here, $\lvlvar, \lvlvarB \in \Slevel$ range over the \emph{sensitivity labels} $\{ \low, \high \}$ that form a lattice with $\low \sqsubseteq \high$.
While any bounded lattice will do, we use the two-element lattice for brevity's sake.

\begin{figure}
\begin{mathpar}
\inferH{typed-if}
  {\typed{\Gamma}{\expr}{\tbool^\low}
   \and \typed{\Gamma}{t}{\type} \and \typed{\Gamma}{u}{\type}}
  {\typed{\Gamma}{{\If \expr then t \Else u}}{\type}}
\and
\inferH{typed-if-flat}
  {\typed{\Gamma}{\expr}{\tbool^\high} \and
   \typed{\Gamma}{v}{\type} \and \typed{\Gamma}{w}{\type} \\
   \mbox{$v$, $w$ are values or variables in $\Gamma$} \and
   \mbox{$\type$ is \flatType}}
  {\typed{\Gamma}{{\If \expr then v \Else w}}{\type}}
\and
\inferH{typed-store}
  {\typed{\Gamma}{\expr}{\tref{\type}}
   \and \typed{\Gamma}{t}{\type}}
  {\typed{\Gamma}{\expr \gets t}{\tunit}}
\and
\inferH{typed-out}
  {\loc \in \Llocset}
  {\typed{\Gamma}{\loc}{\tref{\tint^\low}}}
\end{mathpar}
\caption{A selection of the typing rules.}
\label{fig:typing_rules}
\end{figure}

The typing judgment is of the form $\typed{\Gamma}{\expr}{\type}$, where $\Gamma$ is an assignment of variables to types, $\expr$ is an expression, and $\type$ is a type.
Some typing rules are given in \Cref{fig:typing_rules}, and the rest can be found in \Cref{appendix:types}.
The rule \ruleref{typed-out} shows that every output location $\loc \in \Llocset$ is typed as a reference to a low-sensitivity integer.
By $\type \sqcup \lvlvarB$ we denote the \emph{level stamping}, \eg $\tint^\lvlvar \sqcup \lvlvarB = \tint^{\lvlvar \sqcup \lvlvarB}$.
See \Cref{appendix:types} for the full definition.

To type check the set data structure from \Cref{sec:examples}, we need to support benign branching on high-sensitivity Booleans.
For that purpose we use the rule \ruleref{typed-if-flat}.
In the rule, both branches should either be values or variables, ensuring that they do not perform any computations.
In addition, both branches should be of a \emph{flat type}---$\tint^\high$, $\tbool^\high$, or a product of two flat types.
Function types are not flat because they can leak via timing behavior outside the $\<if>$ branch itself.

Notice that our type system has no sensitivity labels on reference types, and no program counter label on the typing judgment.
While such labels are common in security type systems for languages with (higher-order) references \cite{flowcaml,Terauchi08,RajaniGargModel,ZdancewicThesis}, a
direct adaptation of such type systems is not sound with respect to the termination-sensitive notion of non-interference we consider.
A counterexample is provided in \Cref{sec:discussion:ref_labels}.

\begin{figure}
\begin{align*}
\interp{\tunit}(\val_1, \val_2) \eqdef{} & \val_1 = \val_2 = \unittt \\
\interp{\tint^\lvlvar}(\val_1, \val_2) \eqdef{} & \val_1, \val_2 \in \mathbb{Z} \ast (\lvlvar = \low \to \val_1 = \val_2)\\
\interp{\tbool^\lvlvar}(\val_1, \val_2) \eqdef{} & \val_1, \val_2 \in \mathbb{B} \ast (\lvlvar = \low \to \val_1 = \val_2)\\
\interp{\type \times \type'}(\val_1, \val_2) \eqdef{}&
   \Exists \valB_1, \valB_2, \valB'_1, \valB'_2.\\[-0.2em]
   & \quad \val_1 = (\valB_1, \valB'_1) \ast \val_2 = (\valB_2, \valB'_2) \ast{} \\[-0.2em]
   & \quad \interp{\type}(\valB_1, \valB_2) \ast
   \interp{\type'}(\valB'_1, \valB'_2)\\
\interp{\tref{\type}}(\val_1, \val_2) \eqdef{}&
  \val_1, \val_2 \in \Loc \ast{} \\[-0.4em]
   & \knowInv{\namesp.(\val_1, \val_2)}{
    \begin{array}[t]{@{} l @{}}
    \Exists \valB_1\, \valB_2. \val_1 \mapstoL \valB_1 \ast{}\\
    \quad \val_2 \mapstoR \valB_2 \ast
    \interp{\type}(\valB_1, \valB_2)
    \end{array}} \\
\interp{(\type \to \type')^\lvlvar}(\val_1, \val_2) \eqdef{} &
   \always \big(
     \begin{array}[t]{@{} l @{}}
     \All \valB_1,\valB_2. \interp{\type}(\valB_1, \valB_2) \wand{}\\
     \quad \interpE{\type' \sqcup \lvlvar}{(\val_1\; \valB_1)}{(\val_2\; \valB_2)}\big)
     \end{array}\\
\interpE{\type}(\expr_1, \expr_2) \eqdef{} &
  \dwpre{\expr_1}{\expr_2}{\interp{\type}}
\end{align*}
\caption{The logical relations interpretation of types.}
\label{fig:semtypes}
\end{figure}

\begin{figure}
\begin{mathpar}
\inferH{logrel-if-low}
{\dwpre{\expr_1}{\expr_2}{\interp{\tbool^\low}} \\
\and \dwpre{t_1}{t_2}{\pred}
\and \dwpre{u_1}{u_2}{\pred}}
{\dwpre{\If \expr_1 then t_1 \Else u_1}{\If \expr_2 then t_2 \Else u_2}{\pred}}
\and
\inferH{logrel-store}
{\dwpre{\expr_1}{\expr_2}{\interp{\tref{\type}}}
\and
\dwpre{t_1}{t_2}{\interp{\type}}
}
{\dwpre{(\expr_1 \gets t_1)}
  {(\expr_2 \gets t_2)}
  {\interp{\tunit}}}
% \inferH{logrel-alloc}
% {\dwpre{\expr_1}{\expr_2}{\interp{\type}_{\attklvl}}}
% {\dwpre{\Ref(\expr_1)}{\Ref(\expr_2)}{\interp{\tref{\type}}_{\attklvl}}}
% \and
\end{mathpar}
\caption{A selection of compatibility rules.}
\label{fig:logrel_rules}
\end{figure}

\subsection{Logical relations model}
\label{sec:logrel:model}

We give a semantic model of our type system using logical relations.
The key idea of logical relations is to interpret each type $\type$ as a relation on values, \ie to each type $\type$ we assign an \emph{interpretation} $\interp{\type} : \Val \times \Val \to \Prop$ where $\Prop$ is the type of \thelogic propositions.
Intuitively, $\interp{\type}(\val_1, \val_2)$ expresses that $\val_1$ and $\val_2$ of type $\type$ are indistinguishable by a low-sensitivity attacker.
The definition of $\interp{\type}$ is given in \Cref{fig:semtypes}.
We will now explain some interesting cases in detail.

The interpretation $\interp{\tint^\low}$ contains the pairs of equal integers, while $\interp{\tint^\high}$ contains the pairs of any two integers.
This captures the intuition that a low-sensitivity attacker can observe low-sensitivity integers, but not high-sensitivity integers.

The interpretation $\interp{\tref{\type}}$ captures that references $\loc_1$ and $\loc_2$ are indistinguishable iff they always hold values $\valB_1$ and $\valB_2$ that are indistinguishable at type $\type$.
This is formalized by imposing an invariant that contains both points-to propositions $\loc_1 \mapstoL \valB_1$ and $\loc_2 \mapstoR \valB_2$, as well as the interpretation of $\type$ that links the values $\valB_1$ and $\valB_2$.
Notice that our interpretation of references does not require the locations $\loc_1$ and $\loc_2$ themselves to be syntactically equal.
This is crucial for modeling dynamic allocation (recall that the allocation oracle described in \Cref{sec:lang} may depend on the contents of the heap).

The interpretation $\interp{(\type \to \type')^\lvlvar}$ captures that functions $\val_1$ and $\val_2$ are indistinguishable iff for all inputs $\valB_1$ and $\valB_2$ indistinguishable at type $\type$, the behaviors of the expressions $\val_1\;\valB_1$ and $\val_2\;\valB_2$ are indistinguishable at type $\type' \sqcup \lvlvar$.
To formalize what it means for the behavior of expressions (in this case $\val_1\;\valB_1$ and $\val_2\;\valB_2$) to be indistinguishable, we define the \emph{expression interpretation} $\interpE{\type} : \Expr \times \Expr \to \Prop$ by lifting the value interpretation using double weakest preconditions.

The interpretation of functions is defined using the \emph{persistence modality} $\always$ of Iris~\cite[Section 2.3]{irisJFP}.
Intuitively, $\always \prop$ states that $\prop$ holds without asserting ownership of any non-shareable resources.
Having the persistence modality in this definition % is a technical requirement
is common in logical relations in Iris~\cite{irisIPM}---it ensures that indistinguishable functions remain indistinguishable forever.

The interpretation of expressions $\interpE{\_}$ generalizes to open terms by considering all possible well-typed substitutions.
A (binary) substitution $\gamma$ is a function $\Var \to \Val \times \Val$.
We write $\gamma_i(\expr)$ for a term $\expr$ where each free variable $x$ is substituted by $\pi_i(\gamma(x))$.
We say that a substitution $\gamma$ is well-typed, denoted as $\interp{\Gamma}(\gamma)$, iff  $\All x. \interp{\type}(\gamma(\Gamma(x)))$.
We then define the \emph{semantic typing judgment} as follows:
\[
\typedS{\Gamma}{\expr}{\type} \eqdef
\All \gamma. (\interp{\Gamma}(\gamma) \ast I_{\Llocset}) \wand
  \interpE\type(\gamma_1(\expr), \gamma_2(\expr))
\]
Here, $I_{\Llocset}$ is the invariant on the observable locations (\Cref{sec:soundness_statement}).

\begin{theorem}[Soundness]
\label{thm:adequacy}
If $\typedS{x_1 : \tint^\high, \dots, x_n : \tint^\high}{\expr}{\tint^\low}$ is a derivable in \thelogic, then $\expr$ is secure.
\end{theorem}

\begin{proof}
This a direct consequence of \Cref{thm:soundness1}.
\end{proof}

The \emph{fundamental property} of logical relations states that any program that can be type checked is semantically typed.

\begin{proposition}[Fundamental property]
\label{thm:fundamental}
If $\typed{\Gamma}{\expr}{\type}$, then
$\typedS{\Gamma}{\expr}{\type}$ is derivable in \thelogic.
\end{proposition}

%
%Using the soundness of \thelogic (\Cref{thm:soundness1}) and the fundamental property (\Cref{thm:fundamental}) we arrive at a simple soundness statement for the type system
%
\begin{proof}
This proposition is proved by induction on the typing judgment $\typed{\Gamma}{\expr}{\type}$ using so-called \emph{compatibility rules} for each case.
A selection of these rules is shown in \Cref{fig:logrel_rules}.
\end{proof}

\subsection{Typing via manual proof}
\label{sec:logrel:combine}

When composing the fundamental property (\Cref{thm:fundamental}) and the soundness theorem (\Cref{thm:adequacy}) we obtain that any typed program is secure.
For instance, it allows us to show that the $\<rand>$ program is secure by type checking it, instead of performing a manual proof as done in \Cref{prop:dwp_rand}.

However, semantic typing gives us more---it allows us to combine type-checked code with manually verified code.
Let us consider the examples from \Cref{sec:examples}, which are not typed according to the typing rules, but which we can \emph{prove} to be semantically typed by dropping down to the interpretation of the semantic typing judgment in terms of double weakest preconditions.
\begin{proposition}
\label{prop:dwp_prog2}
$\typedS{}{\mathit{prog}}{\tref{\tint^\low} \to \tint^\high \to \tunit \times \tunit}$.
\end{proposition}
\begin{proof}
This is a direct consequence of \Cref{prop:dwp_prog}.
\end{proof}

\begin{proposition}
${\typedS{}{\mathit{awk}}{\tint^\high \to (\tunit \to \tunit)^\low \to \tint^\low}}$.
\end{proposition}

\begin{proof}
The proposition boils down to showing that for any $i_1, i_2 \in \mathbb{Z}$ and $f_1, f_2$ with $\interp{(\tunit \to \tunit)^\low}(f_1, f_2)$, we have $\dwpre{\mathit{awk}\;i_1\;f_1}{\mathit{awk}\;i_2\;f_2}{\Ret \val_1\, \val_2. \val_1 = \val_2 = 0}$.
We verify this by establishing a monotone protocol similar to the one used in the proof of value-dependent classification in \Cref{sec:logic:example}.
The full proof can be found in the Coq mechanization.
\end{proof}

After establishing the semantic typing for, \eg $\mathit{prog}$ we can use it in any context where a function of the type $\tref{\tint^\low} \to \tint^\high \to \tunit \times \tunit$ is expected.
For example:
\begin{align*}
\typed{h: \tint^\high, f : \tref{\tint^\low} \to \tint^\high \to \tunit \times \tunit \\}{\Let x = \Ref(0) in \Fork{f\ x\ h}; \deref x}{\tint^\low}
\end{align*}
Using the fundamental property (\Cref{thm:fundamental}) we obtain a semantic typing judgment for the above program.
Using \Cref{prop:dwp_prog2} we establish that if we substitute $\mathit{prog}$ for $f$, the resulting program will still be semantically typed, and thus secure by the soundness theorem (\Cref{thm:adequacy}).

The same methodology can be used to assign the types to the safe array operations from \Cref{sec:examples:manual} via manual proof, and compose them with the type checked set data structure from \Cref{sec:set_example}.
The proof can be found in the Coq mechanization.

%%% Local Variables:
%%% mode: latex
%%% TeX-master: "main"
%%% End:

%% file: modular.tex
\section{Modular separation logic specifications}
\label{sec:modular_spec}

Types provide a convenient way to specify program modules, but are not always strong enough to enable the verification of sophisticated clients.
This is particularly relevant if the specification of a program module is to be used in a manual proof or relies on function correctness.
We show that in addition to specifications through types, \thelogic can also be used to prove modular specification in separation logic.
We demonstrate this approach on dynamically created locks (\Cref{sec:locks})
and dynamically classified references (\Cref{sec:modular_value_dep_spec}).

\subsection{Locks}
\label{sec:locks}
The $\HeapLang$ language we consider does not provide locks as primitive constructs, but provides the low-level compare-and-set ($\<CAS>$) operation with which different locking mechanisms can be implemented.
\Cref{fig:lock_rules} displays the implementation and specification of a spin lock.
The specification makes use of a relational generalization of the common \emph{lock predicates} in separation logic~\cite{Biering:2007:BHS:1275497.1275499,CAP,icap}.
The predicate $\isLock{\lkvar_1}{\lkvar_2}{\propC}$ expresses that the pair of locks $\lkvar_1$ and $\lkvar_2$ protect the resources $\propC$, and the predicate $\locked{\lkvar_1}{\lkvar_2}$ expresses that the pair of locks is in acquired state.

\begin{figure}
\paragraph{Implementation of a spin lock}
\begin{align*}
\<let>{} \newlock\ \unittt ={}& \Ref(\<false>) \\
\<let> \<rec>{} \acquire\ \lkvar ={}&
  \<if> \<CAS>(\lkvar,\<false>,\<true>)\ \<then> \unittt \\[-0.2em]
  & \Else \acquire\ \lkvar\\
\<let>{} \release\ \lkvar ={}& \lkvar \gets \<false>
\end{align*}
\paragraph{Modular separation logic specification of locks}
\begin{mathpar}
\inferH{newlock-spec}
{R}
{\dwpre{\newlock\ \unittt}{\newlock\ \unittt}{\Ret \lkvar_1\, \lkvar_2. \isLock{\lkvar_1}{\lkvar_2}{R}}}
\and
\inferH{islock-dup}
{\isLock{\lkvar_1}{\lkvar_2}{\propC}}
{\isLock{\lkvar_1}{\lkvar_2}{\propC} \ast \isLock{\lkvar_1}{\lkvar_2}{R}}
\and
\inferH{acquire-spec}
{\isLock{\lkvar_1}{\lkvar_2}{\propC}}
{\dwpre{\acquire\ \lkvar_1}{\acquire\ \lkvar_2}{\propC \ast \locked{\lkvar_1}{\lkvar_2}}}
\and
\inferH{release-spec}
{\isLock{\lkvar_1}{\lkvar_2}{\propC} \and \propC \and \locked{\lkvar_1}{\lkvar_2}}
{\dwpre{\release\ \lkvar_1}{\release\ \lkvar_2}{\TRUE}}
\end{mathpar}
\caption{Dynamically allocated locks in \thelogic.}
\label{fig:lock_rules}
\label{fig:spin_lock}
\end{figure}

To verify that the spin lock implementation conforms to the lock specification, we define the lock predicates using Iris's mechanism for invariants and user-defined ghost state.
The proof (and invariant) are generalizations of the ordinary proof (and invariant) for functional correctness in Iris.

The rules of our lock specification are similar to the rules in logics with locks as primitives constructs, such as~\cite{CovernMain,SecCSL}.
There are two notable exceptions.
First, in \loccit one needs to fix the set of locks and associated resources upfront,
whereas in \thelogic one can create locks dynamically and attach an arbitrary resource $\propC$ to each lock during the proof.
Second, since locks are not primitive constructs in \thelogic, the specification also applies to different lock implementations, \eg a ticket lock, as we have shown in the Coq mechanization.

\subsection{Dynamically classified references}
\label{sec:modular_value_dep_spec}

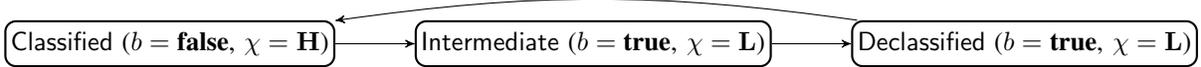
\begin{figure*}
\paragraph{Implementation of dynamically classified references}
\begin{align*}
\<let> \newRec\ \val ={}& \left \{
  \begin{array}{@{} l @{}}
  \mathit{data} = \Ref(\val);\\[-0.2em]
  \mathit{is\_classified} = \Ref(\<false>)
\end{array} \right\} &
  \<let> \classify\ r ={}& r.\mathit{is\_classified} \gets \<true>\\
\<let> \readRec\ r ={}& \deref r.\mathit{data} &
  \<let> \declassify\ r\ \val ={}&
    r.\mathit{data} \gets \val; r.\mathit{is\_classified} \gets \<false> \\
\<let> \storeRec\ r\ \val ={}& r.\mathit{data} \gets \val &
  \<let> \isClassified\ r ={}& \deref r.\mathit{is\_classified}
\end{align*}

\smallskip
\paragraph{Modular separation logic specification of dynamically classified references}

\begin{mathpar}
\inferH{new-vdep}
{\interp{\type \sqcup \lvlvar}(\val_1, \val_2)}
{\dwpre{\newRec\ \val_1}{\newRec\ \val_2}{\Ret r_1\, r_2.
  \valueDependent(\type, r_1, r_2) \ast
  \classification{\lvlvar}{1}}}
\\
\axiomH{valdep-dup}
{\valueDependent(\type, r_1, r_2) \proves \valueDependent(\type, r_1, r_2) \ast \valueDependent(\type, r_1, r_2) }
\ \
\axiomH{class-split}
{\classification{\lvlvar}{q_1} \ast \classification{\lvlvar}{q_2} \provesIff \classification{\lvlvar}{q_1 + q_2}}
\\
\inferH{read-safe}
{\valueDependent(\type, r_1, r_2)}
{\dwpre{\readRec\ r_1}{\readRec\ r_2}{\Ret \val_1\, \val_2. \interp{\type \sqcup \high}(\val_1, \val_2)}}
\and
\inferH{read-seq}
{\valueDependent(\type, r_1, r_2) \and 
\classification{\lvlvar}{q}
}
{\dwpre{\readRec\ r_1}{\readRec\ r_2}{\Ret \val_1\, \val_2. \interp{\type \sqcup \lvlvar}(\val_1, \val_2) \ast \classification{\lvlvar}{q}}}
\and
\inferH{store-safe}
{\valueDependent(\type, r_1, r_2) \and \interp{\type}(\val_1, \val_2)}
{\dwpre{\storeRec\ r_1\ \val_1}{\storeRec\ r_2\ \val_2}
       {\TRUE}}
\and
\inferH{store-seq}
{\valueDependent(\type, r_1, r_2) \and 
\classification{\lvlvar}{q} \and
\interp{\type \sqcup \lvlvar}(\val_1, \val_2)}
{\dwpre{\storeRec\ r_1\ \val_1}{\storeRec\ r_2\ \val_2}
       {\classification{\lvlvar}{q}}}
\and
\inferH{classify-seq}
{\valueDependent(\type, r_1, r_2) \and 
\classification{\lvlvar}{1}}
{\dwpre{\classify\ r_1}{\classify\ r_2}
       {\classification{\high}{1}}}
\and
\inferH{declassify-seq}
{\valueDependent(\type, r_1, r_2) \and 
\classification{\lvlvar}{1} \and
\interp{\type}(\val_1, \val_2)}
{\dwpre{\declassify\ r_1\ \val_1}{\declassify\ r_2\ \val_2}
       {\classification{\low}{1}}}
\and
\inferH{get-classified-seq}
{\valueDependent(\type, r_1, r_2) \and 
\classification{\lvlvar}{q}}
{\dwpre{\isClassified\ r_1}{\isClassified\ r_2}
       {\Ret b_1\, b_2. (b_1 = b_2) \ast \classification{\lvlvar}{q} \ast ((b_1 = \<false>) \to (\lvlvar = \low))}}
\end{mathpar}

\smallskip
\paragraph{The transition system used for the proof}

\begin{center}
\begin{tikzpicture}
\node[state] (classified)
  {$\Classified$ ($b = \<false>$, $\lvlvar = \high$)};
\node[state,right=3em of classified] (intermediate)
  {$\Intermediate$ ($b = \<true>$, $\lvlvar = \low$)};
\node[state,right=3em of intermediate] (declassified)
  {$\Declassified$ ($b = \<true>$, $\lvlvar = \low$)};
\draw  (classified)    edge               (intermediate)
       (intermediate)  edge               (declassified)
       (declassified)  edge[bend right=8] (classified);
;
\end{tikzpicture}
\end{center}
\caption{Dynamically classified references in \thelogic.}
\label{fig:value_dep_derived_specs}
\label{fig:value_dep_lib}
\label{fig:value_dep_sts}
\end{figure*}

We consider a program module
that encapsulates and generalizes dynamically classified references\footnote{In this context declassification refers to changing the dynamic classification of the reference.
It is thus unrelated to static declassification policies~\cite{Sabelfeld:Sands:Declassification}, and the $\declassify$ function is unrelated to the eponymous function from~\cite{Sabelfeld:Myers:DelimitedRelease}.} as used in \Cref{sec:example:classify}.
This program module generalizes to clients with multiple threads and different sharing models.
For example, clients in which multiple threads read and write to the dynamically classified reference, or in which the data gets classified again.
The Coq mechanization contains such an example.
The implementation and specification\footnote{The specification in \Cref{fig:value_dep_lib} is derived from a more general HOCAP-style logically atomic specifications~\cite{HOCAP}, which can be found in \Cref{sec:appendix:modular_specs} and the Coq mechanization.} of the module for dynamically classified references is shown in \Cref{fig:value_dep_lib}.

The main ingredient of the specification is the representation predicate $\valueDependent(\type, r_1, r_2)$,
which expresses that the dynamically classified references $r_1$ and $r_2$ contain related data of type $\type$ at all times.
Since $\valueDependent(\type, r_1, r_2)$ expresses mere knowledge instead of ownership, it is duplicable (\ruleref{valdep-dup}).
With the representation predicate at hand we can formulate weak specifications for some operations.
For instance, the rule \ruleref{read-safe} over-approximates the sensitivity-level of the values returned by the $\readRec$ operation, and dually, the rule \ruleref{store-safe} under-approximates the sensitivity-level of the values stored using the $\storeRec$ operation.
Of course, at times we want to track the precise sensitivity-level.
For that we use a \emph{fractional token} $\classification{\lvlvar}{q}$ with $q \in (0,1]_{\mathbb Q}$.
This token is reminiscent of fractional permissions in separation logic.
The proof rules for $\declassify$ and $\classify$ (\ruleref{declassify-seq} and \ruleref{classify-seq}) require the full fraction ($q = 1$) since they change the classification.
The precise rules for $\readRec$ and $\storeRec$ (\ruleref{read-seq} and \ruleref{store-seq}) do not change the classification, and thus require an arbitrary fraction.
The token is splittable according to \ruleref{class-split} so it can be shared between threads.

Since the rules for $\declassify$ and $\classify$ require a full fraction ($q = 1$), they do not allow for fine-grained sharing\footnote{We can still achieve sharing by storing the token $\classification{\lvlvar}{1}$ in a lock, as outlined in \Cref{sec:locks}.}, \ie they cannot be used to verify a program that runs $\declassify$ in parallel with $\classify$.
It is good that this is impossible---running these operations in parallel results in a race-condition, making it impossible to know what the final classification would be.
However, it \emph{is} possible to verify a program that runs $\declassify$ in parallel with $\readRec$ or $\storeRec$ (using precise rules for these two operations) by sharing the token via an invariant.
To access such a shared token one has to use the more general HOCAP-style logically atomic specifications found in \Cref{sec:appendix:modular_specs} and the Coq mechanization.

\paragraph*{Proof}
In order to verify the implementation, we follow the usual approach of defining the representation predicate $\valueDependent(\type, r_1, r_2)$ and token $\classification{\lvlvar}{q}$ using Iris's invariant and protocol mechanism.
The invariant expresses that, at all times, the fields $\mathit{is\_classified}$ of both records contain the same Boolean value $b$, and that the data in the records are related by $\interp{\type \sqcup \lvlvar}$.
The relation between the Boolean values $b$ and the security label $\lvlvar$, and the way it evolves, are expressed using a protocol visualized as the transition system in \Cref{fig:value_dep_sts}.

%%% Local Variables:
%%% mode: latex
%%% TeX-master: "main"
%%% End:

%% file: soundness.tex
\section{Soundness}
\label{sec:soundness}

To prove soundness of \thelogic (\Cref{thm:adequacy}), we give a model of double weakest preconditions in Iris (\Cref{sec:wp_model}), and then construct a bisimulation out of this model (\Cref{sec:bisim}).
We only give a high-level overview, the details are in \Cref{appendix:soundness}.

\subsection{Model of double weakest preconditions}
\label{sec:wp_model}

The model of the Iris logic~\cite{iris3,irisJFP} consists of three layers:

\begin{itemize}
\item The Iris base logic, which contains the standard separation logic connectives (\eg $*$ and $\wand$), modalities (\eg $\later$, $\always$), and the machinery for user-defined ghost state.
\item The invariant mechanism, which is built as a library on top of the Iris base logic.
\item The Iris program logic, which is built as a library on top of the Iris base logic and invariant mechanism.
  It provides weakest preconditions for proving safety and functional correctness of concurrent programs.
\end{itemize}

We reuse the first two layers of Iris (the base logic and the invariant mechanism), on top of which we model our new notion of double weakest preconditions.
This model is inspired by the model of ordinary (unary) weakest preconditions in Iris and the \emph{product program} construction~\cite{ProductPrograms}.
Intuitively, $\dwpre{\expr_1}{\expr_2}{\pred}$ captures that the expressions $\expr_1$ and $\expr_2$ are executed in lock-step.
This is done by case analysis:

\begin{itemize}
\item Either, both expressions $\expr_1$ and $\expr_2$ are values that are related by the postcondition $\pred$.
\item Otherwise, both expressions $\expr_1$ and $\expr_2$ are reducible, and for any reductions $(\expr_1, \stateS_1) \step (\expr'_1, \stateS'_1)$ and $(\expr_2, \stateS_2) \step (\expr'_2, \stateS'_2)$, the expressions $\expr'_1$ and $\expr'_2$ are still related by $\dwp$.
  If $\expr_1$ and $\expr_2$ fork off threads $\vec{\expr'_1}$ and $\vec{\expr'_2}$, then
  all of the forked-off threads are related pairwise by $\dwp$.
\end{itemize}

\subsection{Constructing a bisimulation}
\label{sec:bisim}

The main challenge of constructing a strong low-bisimulation lies in connecting double weakest preconditions, at the level of separation logic, with strong-low bisimulations, at the meta level.
The construction is done as follows:
\begin{enumerate}
\item We define a relation $\RR$ that ``lifts'' double weakest preconditions out of the \thelogic logic into the meta-level.
\item We then show that the relation $\RR$ satisfies a number of bisimulation-like properties.
\item The relation $\RR$ is not a bisimulation because it is not transitive.
  To fix this, we take its transitive closure $\RRR$. % which is a strong low-bisimulation on configurations.
\end{enumerate}
Finally, we show that the $\dwp$ predicate is sound \wrt the relation $\RR$: if $I_{\Llocset} \vdash \dwpre{\expr}{\exprB}{\Ret \val_1\,\val_2. \val_1 = \val_2}$ can be derived in \thelogic, then $(\expr, \stateS_1)\RR(\exprB, \stateS_2)$ for $\stateS_1 \lowEquiv \stateS_2$.

%%% Local Variables:
%%% mode: latex
%%% TeX-master: "main"
%%% End:

%% file: formalization.tex
\section{Mechanization in Coq}
\label{sec:formalization}

We have mechanized the definition of \thelogic, the type system, the soundness proof,
and all examples and derived constructions in the paper and the appendix in Coq.
The mechanization has been built on top of the mechanization of Iris~\cite{iris2,iris3,irisJFP}, which readily provides the Iris base logic, the invariant mechanism, and the $\HeapLang$ language.

To carry out the mechanization effectively, we have made extensive use of the tactic language MoSeL (formerly Iris Proof Mode) for separation logic in Coq~\cite{irisIPM,MoSeL}.
Using MoSeL we were able to carry out in Coq the typical kind of reasoning steps one would do on paper.
This was essential to mechanize the \thelogic logic (1818 line of Coq code), the type system (1355 lines), and all the examples (3223 lines).

%%% Local Variables:
%%% mode: latex
%%% TeX-master: "main"
%%% End:

%% file: related_work.tex
\section{Related work}
\label{sec:related_work}

\subsection{Security based on strong low-bisimulations}
\label{sec:related_work_bisim}

The security condition we use, a strong low-bisimulation due to Sabelfeld and Sands~\cite{prob-ni}, has been studied in a variety of related work.
In \loccit the notion of a strong low-bisimulation is applied to a first-order stateful language with concurrency.
It is also shown that this notion implies a scheduler-independent bisimulation known as $\rho$-specific probabilistic bisimulation. % and the former can therefore be used to ease the construction of the latter.
Sabelfeld and Sands presented both strong low-bisimulation on thread pools and configurations.
We use the bisimulation relation on configurations because it allows for a flow-sensitive analysis and readily supports dynamic allocation.

Strong low-bisimulations are highly compositional: if a thread $\expr$ is secure \wrt a strong low-bisimulation, then the composition of $\expr$ with \emph{any other} thread is secure.
Unfortunately, this property makes it non-trivial to adapt strong low-bisimulations for analyses that are flow-sensitive in thread composition.
We work around this issue by composing the components at the level of the logic (as double weakest preconditions), and not at the level of the bisimulations, despite the fact that we use strong low-bisimulations as an auxiliary notion in our soundness proof.
By performing the composition at the level of the logic, we can use Iris invariants and modular specifications to put restrictions onto which threads can be composed.

Another way of enabling flow-sensitive analysis was developed by Mantel \etal~\cite{assumptionsguarantees}, who relaxed the notion of a strong low-bisimulation to a \emph{strong low-bisimulation modulo modes}.
Their approach enables rely-guarantee style reasoning at the level of the bisimulation.
Notably, using the notion of strong low-bisimulations modulo modes one can specify that no other threads can read or write to a certain location.

\newcommand{\COVERN}{Covern\xspace}
Based on the notion of strong low-bisimulations modulo modes, the \COVERN project \cite{CovernValueDep,CovernCompiler,CovernMain} developed a series of logics for rely/guarantee reasoning.
Notably, Murray \etal~\cite{CovernMain} presented the first fully mechanized program logic for non-interference of concurrent programs with shared memory, which is also called \COVERN.
While \COVERN is not a separation logic, it has been extended to allow for flexible reasoning about non-interference in presence of value-dependent classifications~\cite{CovernValueDep}.
In terms of the object language, \COVERN does not support fine-grained concurrency, arrays, or dynamically allocated references.
Since \COVERN does not support fine-grained concurrency, locks are modeled as primitives in the language and logic, while they are derived constructs in our work.
As a result of that, \COVERN's notion of strong low-bisimulations is tied to the operational semantics of locks, \ie it is considered \emph{modulo} the variables that are held by locks.
The set of locks, and the variables they protect, has to be provided statically.
Hence their approach does not immediately generalize to support dynamically allocated locks, nor to reason about locks that protect other resources than permissions to write to or read from variables.
Value-dependent classifications are also primitive in \COVERN~\cite{CovernValueDep}, while they are derived constructs in our work.
\COVERN has two separate primitive rules for assignment to ``normal'' variables and for assignment to ``control'' variables (\ie variables that signify the classification levels).

% Another variant of probabilistic bisimulations has been developed by Smith~\cite{weakMarkovBisim}, who considered considered \emph{weak probabilistic bisimulations} for a language with a fork construct, but without dynamic allocation.
% Contrary to our work, the security condition obtained that way is timing-insensitive.

\subsection{Program logics for non-interference}
Early work by Beringer and Hofmann \cite{BeringerHofmannLogic} established a connection between Hoare logic and non-interference.
They did so for a first-order sequential language with a simple non-interference condition.
Non-interference was encoded through self-composition and renaming, making sure that both parts of the composed program operate on different parts of the heap (something that one gets by construction in separation logic).
Notably, they proved the non-interference property of two type systems by %encoding typing derivations into logic, \ie by
constructing models of the type systems in their Hoare logic.
They also showed how to extend their approach to object-oriented type systems.

\subsection{Separation logics for non-interference}
Karbyshev~\etal~\cite{comp-ni-sep} devised a compositional type-and-effect system based on separation logic to prove non-interference of concurrent programs with channels.
Their system is sound \wrt termination-insensitive non-interference allowing for races on low-sensitivity locations.
They consider security for arbitrary (deterministic) schedulers, and allow for a \emph{rescheduling} operation in the programming language to prevent scheduler tainting.
To achieve that, their logical rule for rescheduling treats the scheduler as a splittable separation logic resource, allowing one to share it between threads.
In terms of the object language, they consider a first-order language without dynamic memory allocation, and the concurrency primitives are based on channels with send and receive operations rather than our low-level fine-grained concurrency model.
They do not provide a logic for modular reasoning about program modules.
%(Notably, \cite[Section 5]{comp-ni-sep} contains a detailed study of related work along some axes that we do not consider here.)

The recently proposed separation logic SecCSL~\cite{SecCSL} enables reasoning about value-dependent information flow control policies through a relational interpretation of separation logic.
One of the main advantages of the SecCSL approach is its amenability to automation.
However, to achieve that, they restrict to a first-order separation logic with restricted language features, \ie a first-order language with first-order references, and a coarse-grained synchronization mechanism.
SecCSL does not support dynamically allocated references out of the box.
However, we believe that it can be extended to support dynamic allocation, as long as the semantics for allocation are deterministic and do not depend on the global heap.

The security condition in SecCSL~\cite{SecCSL} is non-standard, and is geared to providing meaning to the intermediate Hoare triples.
Because of that, their formulation of non-interference is closely intertwined with the semantics of the logic.

Costanzo and Shao~\cite{sepLogic:declass:2014} devised a separation logic for proving non-interference of first-order sequential programs.
One of the novelties of their system is the support for declassification in the form of \emph{delimited release}~\cite{Sabelfeld:Myers:DelimitedRelease}.
While we do not study declassification policies in this paper, we believe that the approach of Costanzo and Shao can be adapted to our setting, provided that we are willing to relax the notion of a strong low-bisimulation.

\subsection{Type systems for non-interference}
As discussed in the introduction (\Cref{sec:intro}), a lot of work on non-interference in the programming languages area has focused on type-system based approaches.
Such approaches are amendable to high degrees of automation, but lack the ability to reason about functional correctness.
Due to an abundance of prior work on in this area, we restrict to directly related work.

Pottier and Simonet developed Flow Caml \cite{flowcaml}, a type system for termination-insensitive non-interference for sequential higher-order language in the spirit of Caml.
Soundness \wrt non-interference is proven with the \emph{product programs} technique.
This kind of self-composition was an inspiration for our model of double weakest preconditions, although we avoid self-composition of programs at the syntactic level.

Terauchi \cite{Terauchi08} devised a capabilities-based type system for \emph{observational determinism} \cite{ObsDet}.
Observational determinism is a formulation of non-interference for concurrent programs that is substantially different from the strong low-bisimulation considered in this paper.
In particular, under observational determinism, no races on low-sensitivity locations are allowed, ruling out \eg the $\<rand>$ function from \Cref{sec:rand_example}.
% See \cite[Section 5]{comp-ni-sep} for detailed discussion.

\subsection{Logical relation models}

The technique of logical relations is widely used for proving the soundness of type systems and logics.
The work on step-indexing~\cite{Ahmed:StepIndexed,Appel:ModalModel} made it possible to scale logical relations to languages with higher-order references and recursive types.
Notably, Rajani and Garg \cite{RajaniGargModel} describe a step-indexed Kripke-style model for two information flow aware type systems for a sequential language with higher-order references.
While they do not consider concurrency and their notion of non-interference is different from ours (their notion is termination- and progress-insensitive), their model is similar in spirit.
However, we make use of the ``logical'' approach to step-indexing~\cite{lslr} in Iris to avoid explicit step-indexes in definitions and proofs.

The relational model of our type system is directly inspired by a line of work on interpretation of type systems and logical relations in Iris \cite{irisIPM,irisEffects,RustBelt,irisRunSt,reloc}, but this previous work focused on reasoning about safety and contextual equivalence of programs, while we target non-interference.
For that purpose we developed double weakest preconditions.

The idea of using logical relations to reason about the combination of typed and manually verified code has been used before in the context of Iris.
Notably, Jung \etal~\cite{RustBelt} use it to reason about unsafe code in Rust, and Krogh-Jespersen \etal~\cite{irisEffects} use it in the context of type-and-effect systems.

%%% Local Variables:
%%% mode: latex
%%% TeX-master: "main"
%%% End:

%% file: conclusion.tex
\section{Conclusions and future work}

We have presented \thelogic---the first separation logic for non-interference that combines type checking and manual proof.
It supports fine-grained concurrency, higher-order functions, and dynamic (higher-order) references.
The key feature of \thelogic is its novel connective for double weakest preconditions, which in combination with Iris-style invariants, allows for compositional reasoning.
We have proved soundness of \thelogic with respect to a standard notion of security.

In future work we want to develop a more expressive type system.
To develop such a type system, we want to transfer back reasoning principles from \thelogic into constructs that can be type checked automatically.
Moreover, we would like to study declassification in the sense of delimited information release and static declassification policies~\cite{Sabelfeld:Sands:Declassification,Sabelfeld:Myers:DelimitedRelease,Banerjee:2007:declass,sepLogic:declass:2014}.

%%% Local Variables:
%%% mode: latex
%%% TeX-master: "main"
%%% End:

%% file: appendix.tex
\section{Type system}
\label{appendix:types}

\subsection{Typing rules}

\begin{figure*}
\paragraph{Subtyping rules}\hfill

\vspace{-2em}
\begin{mathpar}
\infer{}
{\type <: \type}
\and
\infer{\type_1 <: \type_2 \and \type_2 <: \type_3}
{\type_1 <: \type_3}
\vspace{-0.5em} \\
\infer{\lvlvar_1 \sqsubseteq \lvlvar_2}
{\tint^{\lvlvar_1} <: \tint^{\lvlvar_2}}
\and
\infer{\lvlvar_1 \sqsubseteq \lvlvar_2}
{\tbool^{\lvlvar_1} <: \tbool^{\lvlvar_2}}
\and
\infer
{\lvlvar_1 \sqsubseteq \lvlvar_2
\and \type'_1 <: \type_1
\and \type_2 <: \type'_2}
{(\type_1 \to \type_2)^{\lvlvar_1} <: (\type'_1 \to \type'_2)^{\lvlvar_2}}
\and
\infer
{\type_1 <: \type'_1
\and \type_2 <: \type'_2}
{\type_1 \times \type_2 <: \type'_1 \times \type'_2}
\end{mathpar}

\medskip
\paragraph{Typing rules}
\begin{mathpar}
\infer
  {\type <: \type'
   \and \typed{\Gamma{}}{\expr}{\type{}}}
  {\typed{\Gamma{}}{\expr}{\type'{}}}
\and
\infer
  {\Gamma(x) = \type}
  {\typed{\Gamma{}}{x}{\type}}
\and
\infer
  {\loc \in \Llocset}
  {\typed{\Gamma}{\loc}{\tref{\tint^\low}}}
\and
\infer
  {}{\typed{\Gamma}{\unittt}{\tunit}}
\and
\infer
  {i \in \mathbb{Z}}
  {\typed{\Gamma}{i}{\tint^\lvlvar}}
\and
\infer
  {b \in \mathbb{B}}
  {\typed{\Gamma}{b}{\tbool^\lvlvar}}
\and
\infer
  {\typed{\Gamma}{\expr}{\tint^\lvlvar} \and \typed{\Gamma}{\exprB}{\tint^\lvlvarB}}
  {\typed{\Gamma}{\expr + \exprB}{\tint^{\lvlvar \sqcup\lvlvarB}}}
  \and
  \infer
  {\typed{f : (\type \to \type')^\lvlvar, x : \type, \Gamma}
    {\expr}{\type' \sqcup \lvlvar}}
  {\typed{\Gamma}
    {\Rec f x = \expr}{(\type \to \type')^\lvlvar}}
\and
\infer
  {\typed{\Gamma}{\expr}{(\tau \to \tau')^\lvlvar} \and
  \typed{\Gamma}{\exprB}{\tau}
  }
  {\typed{\Gamma}{\expr\; \exprB}{\tau' \sqcup \lvlvar}}
\and
\infer
  {\typed{\Gamma}{\expr}{\tbool^\low}
   \and \typed{\Gamma}{\expr_1}{\type}
   \and \typed{\Gamma}{\expr_2}{\type}}
  {\typed{\Gamma}{\If \expr then \expr_1 \Else \expr_2}{\type}}
\and
\infer
  {\mbox{$\begin{array}[t]{@{} c}
   \typed{\Gamma}{\expr}{\tbool^\high} \quad
   \typed{\Gamma}{v}{\type} \quad \typed{\Gamma}{w}{\type} \\[-0.2em]
   \mbox{$v$, $w$ are values or variables in $\Gamma$} \quad
   \mbox{$\type$ is \flatType}
   \end{array}$}}
  {\typed{\Gamma}{{\If \expr then v \Else w}}{\type}}
\and
\infer
  {\typed{\Gamma}{\expr}{\type}}
  {\typed{\Gamma}{\Fork{\expr}}{\tunit}}
\and
\infer
  {\typed{\Gamma}{\expr}{\type}}
  {\typed{\Gamma}{\Ref(\expr)}{\tref{\type}}}
\and
\infer
  {\typed{\Gamma}{\expr}{\tref{\type}}}
  {\typed{\Gamma}{\deref \expr}{\type}}
\and
\infer
  {\typed{\Gamma}{\expr_1}{\tref{\type}}
   \and \typed{\Gamma}{\expr_2}{\type}}
  {\typed{\Gamma}{\expr_1 \gets \expr_2}{\tunit}}
\and
\infer
  {\typed{\Gamma{}}{\expr_1}{\tref{\tint^\lvlvar}}
   \and \typed{\Gamma{}}{\expr_2}{\tint^\lvlvar}}
  {\typed{\Gamma{}}{\<FAA>(\expr_1, \expr_2)}{\tint^\lvlvar}}
\end{mathpar}
\caption{Typing rules of the \thelogic type system.}
\label{fig:type_system}
\end{figure*}
\begin{figure*}
% \paragraph{Compatibility rules} \hfill
% \vspace{-0.5em}
\begin{mathparpagebreakable}
\inferH{interp-sub}
{\type_1 <: \type_2 \and
\interp{\type_1}(\val_1, \val_2)}
{\interp{\type_2}(\val_1, \val_2)}
\and
\inferH{logrel-sub}
{\type_1 <: \type_2 \and
  \dwpre{\expr_1}{\expr_2}{\interp{\type_1}}}
{\dwpre{\expr_1}{\expr_2}{\interp{\type_2}}}
\and
\inferH{logrel-int-low}
{i \in \mathbb{Z}}
{\dwpre{i}{i}{\interp{\tint^\lvlvar}}}
\and
\inferH{logrel-int}
{i_1, i_2 \in \mathbb{Z} \and \lvlvar \not\sqsubseteq \low}
{\dwpre{i_1}{i_2}{\interp{\tint^\lvlvar}}}
\and
\inferH{logrel-bool-low}
{b \in \mathbb{B}}
{\dwpre{b}{b}{\interp{\tbool^\lvlvar}}}
\and
\inferH{logrel-bool}
{b_1, b_2 \in \mathbb{B} \and \lvlvar \not\sqsubseteq \low}
{\dwpre{b_1}{b_2}{\interp{\tbool^\lvlvar}}}
\and
\inferH{logrel-binop}
{\dwpre{\expr_1}{\expr_2}{\interp{\tint^{\lvlvar}}}
\and
\dwpre{\exprB_1}{\exprB_2}{\interp{\tint^{\lvlvarB}}}}
{\dwpre{\expr_1 + \exprB_1}{\expr_2 + \exprB_2}
  {\interp{\tint^{\lvlvar \sqcup \lvlvarB}}}}
\and
\inferH{logrel-rec}
{\always \All f_1\, f_2\, \val_1\, \val_2.
  \interp{(\type_1 \to \type_2)^\lvlvar}(f_1, f_2) \ast \interp{\type_1}(\val_1, \val_2) \wand
  \dwpre{\subst{\subst{\expr_1}{x}{\val_1}}{f}{f_1}}%
        {\subst{\subst{\expr_2}{x}{\val_2}}{f}{f_2}}%
        {\interp{\type_2 \sqcup \lvlvar}}
}
{\dwpre{(\Rec f x := \expr_1)}{(\Rec f x := \expr_2)}{\interp{(\type_1 \to \type_2)^\lvlvar}}}
\and
\inferH{logrel-app}
{\dwpre{\expr_1}{\expr_2}{\interp{(\type_1 \to \type_2)^\lvlvar}}
\and
\dwpre{\exprB_1}{\exprB_2}{\interp{\type_1}}}
{\dwpre{\expr_1\ \exprB_1}{\expr_2\ \exprB_2}{\interp{\type_2 \sqcup \lvlvar}}}
\and
\inferH{logrel-if}
{\dwpre{\expr_1}{\expr_2}{\interp{\tbool^\lvlvar}}\and
\dwpre{t_1}{t_2}{\pred} \wedge
\dwpre{u_1}{u_2}{\pred}
\wedge (\lvlvar \not\sqsubseteq \low \to
   (\dwpre{u_1}{t_2}{\pred} \wedge
   \dwpre{t_1}{u_2}{\pred}))}
{
  \dwpre{(\If \expr_1 then t_1 \Else u_1)}
  {(\If \expr_2 then t_2 \Else u_2)}{\pred}
}
\and
\inferhref{logrel-if-low}{logrel-if-low'}
{\dwpre{\expr_1}{\expr_2}{\interp{\tbool^\low}}
\and \dwpre{t_1}{t_2}{\pred}
\and \dwpre{u_1}{u_2}{\pred}}
{\dwpre{\If \expr_1 then t_1 \Else u_1}{\If \expr_2 then t_2 \Else u_2}{\pred}}
\and
\inferhref{logrel-if-flat}{logrel-if-flat'}
{\dwpre{\expr_1}{\expr_2}{\interp{\tbool^\high}}
\and \dwpre{v_1}{v_2}{\interp{\type}}
\and \dwpre{w_1}{w_2}{\interp{\type}}
\and \mbox{$\type$ is \flatType}}
{\dwpre{\If \expr_1 then v_1 \Else w_1}{\If \expr_2 then v_2 \Else w_2}{\interp{\type}}}
\and
\inferH{logrel-fork}
{\dwpre{\expr_1}{\expr_2}{\pred}}
{\dwpre{\Fork{\expr_1}}{\Fork{\expr_2}}{\interp{\tunit}}}
\and
\inferH{logrel-alloc}
{\dwpre{\expr_1}{\expr_2}{\interp{\type}}}
{\dwpre{\Ref(\expr_1)}{\Ref(\expr_2)}{\interp{\tref{\type}}}}
\and
\inferH{logrel-load}
{\dwpre{\expr_1}{\expr_2}{\interp{\tref{\type}}}}
{\dwpre{\deref \expr_1}{\deref \expr_2}{\interp{\type}}}
\and
\inferhref{logrel-store}{logrel-store'}
{\dwpre{\expr_1}{\expr_2}{\interp{\tref{\type}}}
\and
\dwpre{t_1}{t_2}{\interp{\type}}}
{\dwpre{(\expr_1 \gets t_1)}
  {(\expr_2 \gets t_2)}
  {\interp{\tunit}}}
\and
\inferH{logrel-faa}
{\dwpre{\expr_1}{\expr_2}{\interp{\tref{\tint^\lvlvar}}}
\and
\dwpre{t_1}{t_2}{\interp{\tint^\lvlvar}}
}
{\dwpre{\<FAA>(\expr_1, t_1)}{\<FAA>(\expr_2, t_2)}{\interp{\tint^\lvlvar}}}
\end{mathparpagebreakable}
\caption{Compatibility rules of \thelogic type system.}
\label{fig:compat_rules_appendix}
\end{figure*}
The subtyping and typing rules can be found in \Cref{fig:type_system};
the compatibility rules can be found in \Cref{fig:compat_rules_appendix}.
The \emph{level stamping} function is defined as:
\begin{align*}
\tunit \sqcup \lvlvarB \eqdef{}& \tunit &
  (\type \times \type') \sqcup \lvlvarB \eqdef{}& (\type \sqcup \lvlvarB) \times (\type' \sqcup \lvlvarB) \\
\tint^\lvlvar \sqcup \lvlvarB \eqdef{}& \tint^{\lvlvar \sqcup \lvlvarB} &
  (\tref{\type}) \sqcup \lvlvarB \eqdef{}&\tref{\type} \\
\tbool^\lvlvar \sqcup \lvlvarB \eqdef{}& \tbool^{\lvlvar \sqcup \lvlvarB} &
  (\type \to \type')^\lvlvar \sqcup\lvlvarB \eqdef{}& (\type \to \type')^{\lvlvar \sqcup \lvlvarB}
\end{align*}
\emph{Flat types} are defined inductively as:
\begin{mathpar}
\axiom
{\tunit \mbox{ is \flatType}}
\and
\axiom
{\tint^\high \mbox{ is \flatType}}
\and
\axiom
{\tbool^\high \mbox{ is \flatType}}
\and
\infer
{\type_1 \mbox{ is \flatType}
\and \type_2 \mbox{ is \flatType}}
{\type_1 \times \type_2 \mbox{ is \flatType}}
\end{mathpar}

\subsection{Sensitivity labels on references and aliasing}
\label{sec:discussion:ref_labels}

Most type systems for non-interference for languages with (higher-order) references annotate reference types with sensitivity labels, and annotate the typing judgment with a \emph{program counter} label \cite{flowcaml,Terauchi08,RajaniGargModel,ZdancewicThesis}.
These annotations are used to prevent leaks via aliasing, while allowing more programs to be typed.
Our type system (\Cref{sec:logrel}) does not have such annotations because some programs that are typeable using such annotations are not secure \wrt a termination-sensitive notion of non-interference (\eg strong low-bisimulation).
For example, termination-insensitive type systems usually accept the following program as secure: 
\[
  (\If h then f \Else g)\; \unittt
\]
Here, $h$ is a high-sensitivity Boolean, and $f$ and $g$ are functions of type $(\tunit \to \tunit)^\low$.
Under a termination-sensitive notion of security, the program is not secure because $f$ and $g$ can examine different termination behavior.

Despite this, let us examine why exactly we do not need labels on reference types to prevent leaks via aliasing, and argue that our approach still allows for benign aliasing of references.
A classic example of an information leak via aliasing is:
\begin{align*}
\<let> p_1\;r\;s\;h ={}
& r \gets \<true>; s \gets \<true>;\\
& \Let x = (\If h then r \Else s) in\\
& x \gets \<false>; \deref r
\end{align*}
Both $r$ and $s$ contain low-sensitivity data, but by aliasing one or the other with $x$, the program leaks the high-sensitivity value $h$.
In previous approaches such leaks are avoided by tracking aliasing information through sensitivity labels on references.
The variable $x$ would be typed as $(\tref{\tint^\low})^\high$ because it was aliased in a high-sensitivity context (branching on $h$).
The consequent assignment $x \gets \<false>$ is then prevented by the type system since the label on the reference ($\high$) is not a below the label of the values that are stored in the reference ($\low$).

In \thelogic, the variable $x$ will not be typeable at all.
To see why that is the case, suppose we want to prove that the program is secure.
For this, we let $h_1$ and $h_2$ denote high-sensitivity inputs for two runs of the program, and $r_1, s_1$  (resp.~$r_2, s_2$) denote the low-sensitivity references arguments for the left-hand side program (resp.~right-hand side program).
Under these high-sensitivity inputs, we need to prove that the bodies of the let-expressions are indistinguishable, \ie
\[
  \dwpre{\If h_1 then r_1 \Else s_1}{\If h_2 then r_2 \Else s_2}{\interp{\tref{\tint^\low}}}
\]
Proving this proposition, would in particular require proving $\dwpre{r_1}{s_2}{\interp{\tref{\tint^\low}}}$, which is impossible in \thelogic.

If we remove the trailing assignment $x \gets \<false>$ the resulting program $p_2$ becomes trivially secure, and many termination-insensitive type systems accept it as such:
\begin{align*}
\<let> p_2\;r\;s\;h ={}
& r \gets \<true>; s \gets \<true>;\\
& \Let x = (\If h then r \Else s) in\\
& \deref r
\end{align*}
Our type system cannot be used to type check this example: as we have just explained, we cannot type the let-expression at all.
Despite this, we can fall back on the double weakest preconditions to verify the security of $p_2$, \ie we can prove:
\begin{align*}
& \interp{\tref{\tbool^\low}}(r_1, r_2) \ast
\interp{\tref{\tbool^\low}}(s_1, s_2) \ast{} \\
& \;\;\interp{\tbool^\high}(h_1, h_2) \vdash
   \dwpre{p_2\;r_1\;s_1\;h_1}{p_2\;r_2\;s_2\;h_2}{\interp{\tunit}}
\end{align*}
by symbolic execution.
Using our logic, we can perform a case distinction on the Boolean values $h_1$ and $h_2$,
which amounts to proving $\dwpre{p_2\;r_1\;s_1\;\<true>}{p_2\;r_2\;s_2\;\<true>}{\interp{\tunit}}$, $\dwpre{p_2\;r_1\;s_1\;\<true>}{p_2\;r_2\;s_2\;\<false>}{\interp{\tunit}}$, \etc{}
We solve all these goals by symbolic execution.
This example demonstrates the advantages of combining typing with manual proofs.

We believe that the restriction on the typing of the $\<let> x$-binding is not unreasonable in case of termination-sensitive and progress-sensitive security condition.
As we have mentioned, if we take termination and timing behavior into account, the liberal compositional reasoning that is enjoyed by termination-insensitive type systems is no longer sound.
In presence of higher-order functions and store, we can write the counterexample from the beginning of this section in the form of $p_2$ to obtain the program $p_3$ below:
\begin{align*}
\<let> p_3\;f\;g\;h ={}
& r \gets f; s \gets g;\\
& \Let x = (\If h then r \Else s) in\\
& (\deref x)\unittt
\end{align*}
The variable $x$ now aliases a reference to a function.
If $f$ and $g$ exhibit different termination behavior, then the value of $h$ can be observed by invoking $\deref x$.

\begin{figure*}[t!]
\begin{mathpar}
\inferH{valdep-persistent}
{\valueDependent(\type, r_1, r_2)}
{\always \valueDependent(\type, r_1, r_2)}
\and
\inferH{classification-agree}
{\classification{\lvlvar_1}{q_1}
\and \classification{\lvlvar_2}{q_2}
}
{\lvlvar_1 = \lvlvar_2}
\and
\inferH{classification-op}
{}
{\classification{\lvlvar}{q_1}
\ast \classification{\lvlvar}{q_2} \dashv\vdash \classification{\lvlvar}{q_1+q_2}}
\and
\inferH{classification-1-exclusive}
{\classification{\lvlvar}{1} \and\classification{\lvlvar}{q}}
{\FALSE}
\and
\inferH{classification-auth-agreee}
{\classificationAuth{\lvlvar_1} \and
\classification{\lvlvar_2}{q}}
{\lvlvar_1 = \lvlvar_2}
\and
\inferH{classification-update}
{\classificationAuth{\lvlvar} \and
\classification{\lvlvar}{1}}
{\pvs \classificationAuth{\lvlvar'} \ast \classification{\lvlvar'}{1}}
\and
\inferH{read-spec}
{\valueDependent(\type, r_1, r_2) \and
(\All \lvlvar\, \val_1\, \val_2.
  \classificationAuth{\lvlvar} \ast \interp{\type \sqcup \lvlvar}(\val_1, \val_2) \vsW \classificationAuth{\lvlvar} \ast \pred(\val_1, \val_2))
}
{\dwpre{\readRec\ r_1}{\readRec\ r_2}{\pred}}
\and
\inferH{write-spec}
{\valueDependent(\type, r_1, r_2) \and
(\All \lvlvar.
  \classificationAuth{\lvlvar} \vsW \classificationAuth{\lvlvar}
\ast \interp{\type \sqcup \lvlvar}(\val_1, \val_2)
 \ast \pred(\unittt, \unittt))}
{\dwpre{\storeRec\ r_1\ \val_1}{\storeRec\ r_2\ \val_2}{\pred}}
\and
\inferH{is-classified-spec}
{\valueDependent(\type, r_1, r_2) \and
(\All \lvlvar\, b.   \classificationAuth{\lvlvar} \vsW \classificationAuth{\lvlvar}
\ast ((b = \<false> \to \lvlvar = \low) \wand \pred(b, b)))}
{{\dwpre{\isClassified\ r_1}{\isClassified\ r_2}{\pred}}}
\and
\inferH{declassify-spec}
{\valueDependent(\type, r_1, r_2) \and
\classification{\lvlvar}{q} \and\\
(\classificationAuth{\lvlvar} \ast \classification{\lvlvar}{q}
\vsW \classificationAuth{\low} \ast \classification{\low}{q}
\ast (\classification{\low}{q} \wand \pred(\unittt, \unittt)))
}
{\dwpre{\declassify\ r_1\ \val_1}{\declassify\ r_2\ \val_2}{\pred}}
\and
\inferH{classify-spec}
{\valueDependent(\type, r_1, r_2) \and
\classification{\lvlvar}{q} \and\\
(\classificationAuth{\lvlvar} \ast \classification{\lvlvar}{q}
\vsW \classificationAuth{\high} \ast \pred(\unittt, \unittt)}
{\dwpre{\classify\ r_1}{\classify\ r_2}{\pred}}
\and
\inferH{new-vdep-spec}
{\interp{\type \sqcup \lvlvar}(\val_1, \val_2) \and
(\All r_1\, r_2. \valueDependent(\type, r_1, r_2) \ast
  \classification{\lvlvar}{1} \wand \pred(r_1, r_2))}
{\dwpre{\newRec\ \val_1}{\newRec\ \val_2}{\pred}}
\end{mathpar}
\caption{HOCAP-style specifications for dynamically classified references.}
\label{fig:dyn_class_ref_rules}
\end{figure*}

\subsection{Generalized rule for branching}
\label{sec:discussion:logrel_if}
\begin{comment}
\begin{figure*}[t]
\begin{mathpar}
\inferH{logrel-if}
{\dwpre{\expr_1}{\expr_2}{\interp{\tbool^\lvlvar}}\and \\
%
\dwpre{t_1}{t_2}{\pred}
\wedge (\lvlvar \not\sqsubseteq \low \to
   \dwpre{u_1}{t_2}{\pred})
\wedge
\dwpre{u_1}{u_2}{\pred}
\wedge (\lvlvar \not\sqsubseteq \low \to
   \dwpre{t_1}{u_2}{\pred})}
{
  \dwpre{(\If \expr_1 then t_1 \Else u_1)}
  {(\If \expr_2 then t_2 \Else u_2)}{\pred}
}
\end{mathpar}
\caption{Generalized if-rule.}
\label{fig:logrel_if_general}
\end{figure*}
\end{comment}
The notion of security that we use (strong low-bisimulation) allows for branching on high-sensitivity data, provided that the timing behavior of the branches is indistinguishable.
However, if we branch on a high-sensitivity Boolean, it is insufficient to verify that each individual branch is secure, we also have to verify that the two different branches are indistinguishable for the attacker.
This kind of condition is present in the rule \ruleref{logrel-if} in \Cref{fig:type_system}.
We can speak of two different branches being indistinguishable because we have moved from a unary typing system to a binary logic.

Recall, that our inference rules are interpreted as an separating implication, where the premises are joined together by a separating conjunction.
To prove each premise, the user of the rule has to distribute the resources they currently have among the premises.
The last four premises in \ruleref{logrel-if}, however, are joined by a regular intuitionistic conjunction ($\wedge$).
The user still has to prove both of those premises if they wish to apply the rule, but this time they do not have to split their resources, \ie they are able to reuse the same resource to prove all the premises.
This corresponds to the fact that there are four possible combinations of branches, but only one of the combinations can actually occur.

\section{HOCAP-style modular specifications}
\label{sec:appendix:modular_specs}

We provide modular logically atomic specifications for the module of dynamically classified references (\Cref{sec:modular_value_dep_spec}) in \Cref{fig:dyn_class_ref_rules}.
These specifications are stronger than the ones given in \Cref{fig:dyn_class_ref_rules} in the sense that they are \emph{logically atomic}, \ie they allow one to open invariants around operations.
This is achieved using the HOCAP~\cite{HOCAP} approach to logical atomicity.
More information about HOCAP-style specifications in Iris can be found in~\cite[Chapter 10]{irisLN}.
Note that the weaker specifications in \Cref{fig:value_dep_derived_specs} (from \Cref{sec:modular_value_dep_spec})
can be derived from the HOCAP-style specifications in \Cref{fig:dyn_class_ref_rules}.

\section{Soundness}
\label{appendix:soundness}
\begin{figure*}[t!]
\begin{align*}
\dwpre{\expr_1}{\expr_2}{\pred} \eqdef{}&
  \begin{cases}
  \pvs[\top] \pred(\expr_1, \expr_2)
  & \hfill\textnormal{if $\expr_1,\expr_2 \in \Val$} \\
  \pvs[\top] \FALSE
  & \hspace{-3em}\hfill\textnormal{if $\expr_1 \in \Val$ xor $\expr_2 \in \Val$} \\
  \All \stateS_1\, \stateS_2. \staterel(\stateS_1, \stateS_2)
    \wand \pvs[\top][\emptyset]\ \red(\expr_1, \stateS_1) \ast \red(\expr_2, \stateS_2) \ast{}  \\
  \quad \All \expr'_1\, \stateS'_1\, \vec{\expr_1}\, \expr'_2\, \stateS'_2\, \vec{\expr_2}.
    (\expr_1, \stateS_1) \step  (\expr'_1\vec{\expr_1}, \stateS'_1) \wedge{} (\expr_2, \stateS_2) \step  (\expr'_2\vec{\expr_2}, \stateS'_2) \wand{}\\
  \quad \quad 
    \pvs[\emptyset][\emptyset] \later \pvs[\emptyset][\top] \staterel(\stateS'_1, \stateS'_2) \ast \dwpre{\expr'_1}{\expr'_2}{\pred} \ast
    \Sep_{(\expr''_1, \expr''_2) \in \vec{\expr_1} \times \vec{\expr_2}}. \dwpre{\expr''_1}{\expr''_2}{\TRUE}
  & \hfill \textnormal{otherwise}
  \end{cases}
\end{align*}
\caption{The model of double weakest preconditions.}
\label{fig:dwp_def}
\end{figure*}

%We give the formal model of double weakest preconditions in Iris, and the proof of the construction of a bisimulation.
%For reasons of space, we cannot explain the details of Iris, so we presume the reader is familiar with those.

\Cref{fig:dwp_def} contains the formal model of $\dwpre{\expr_1}{\expr_2}{\pred}$ as a definition in the Iris framework.
Recall that this definition captures that the expressions $\expr_1$ and $\expr_2$ are executed in a lock-step manner.
Since this definition is inspired by the definition of ordinary weakest preconditions in Iris and the \emph{product program} construction~\cite{ProductPrograms}, instead of Iris's \emph{state interpretation} $\stateinterp : \State \to \Prop$, we have a \emph{state relation} $\staterel : \State\times\State \to \Prop$ that keeps track of both the left and right-hand side heaps.

We now provide the details of the construction of a strong low-bisimulation out of a double weakest precondition proof.

\begin{definition}
\label{def:R_relation}
We define the relation $\RR$ on configurations of the same size to be the following:
\begin{align*}
(\expr_0 \expr_1 \dots \expr_m, \stateS_1)
  \RR (\exprB_0 \exprB_1 \dots \exprB_m, \stateS_2) \eqdef{}& \Exists n:\nat.\\
& \hspace{-15em}
\begin{array}{l}
\TRUE \proves
  \left(\!\pvs[\top][\emptyset] \later \pvs[\emptyset][\top]\!\right)^n \pvs[\top]
    \staterel(\stateS_1, \stateS_2) \ast I_{\Llocset} \ast{} \\[-0.2em]
  \qquad \dwpre{\expr_0}{\exprB_0}{\Ret \val_1\, \val_2. \val_1 = \val_2} \ast{} \\
  \qquad \Sep_{1 \leq i \leq m}.\ \dwpre{\expr_i}{\exprB_i}{\TRUE}
\end{array}
\end{align*}
%
%Here, $I_{\Llocset}$ is the invariant defined in \Cref{sec:soundness_statement}.
%Note that we require the thread pools to be of equal size.
\end{definition}

\balance

Note that $\RR$ is defined at the meta-level, \ie outside \thelogic; in particular the existential quantifier $\exists n:\nat$ is at the meta-level.
The relation $\RR$ relates two configurations if all the threads are related by a double weakest precondition, and execution of the main threads furthermore result in the same value.
The invariant $I_{\Llocset}$ (which has been defined in \Cref{sec:soundness_statement}) guarantees that the output locations $\Llocset$ always contain the same data between any executions of the two configurations.
The existentially quantified natural number $n$ bounds the number of times the definition of double weakest preconditions has been unfolded.
It is needed to show that $\RR$ is closed under reductions.

The relation $\RR$ allows one to ``lift'' double weakest precondition proofs from inside the logic:

\begin{proposition}
\label{prop:dwp_lift_bisim}
If $I_{\Llocset} \vdash \dwpre{\expr}{\exprB}{\Ret \val_1\,\val_2. \val_1 = \val_2}$
is derivable in \thelogic, then $(\expr, \stateS_1)\RR(\exprB, \stateS_2)$ for any \mbox{$\stateS_1 \lowEquiv \stateS_2$}.
\end{proposition}

\begin{proof}
For showing $(\expr, \stateS_1)\RR(\exprB, \stateS_2)$, pick $n=0$.
Because $\stateS_1$ and $\stateS_2$ agree on the $\Llocset$-locations (\ie $\stateS_1 \lowEquiv \stateS_2$), we can establish the state relation $\staterel(\stateS_1, \stateS_2)$ and the invariant $I_{\Llocset}$.
\end{proof}

\begin{lemma}
\label{lem:dwp_rel}
The following properties hold:
\begin{enumerate}
\item $\RR$ is symmetric;
\item If $(\val\vec{\expr}, \stateS_1) \RR (\valB\vec{\exprB}, \stateS_2)$, then $\val = \valB$;
\item If $(\vec{\expr}, \stateS_1) \RR (\vec{\exprB}, \stateS_2)$, then $\lvert \vec{\expr} \rvert = \lvert \vec{\exprB}\rvert$ and $\stateS_1 \lowEquiv \stateS_2$;
\item If $(\expr_0 \dots \expr_i \dots, \stateS_1) \RR (\exprB_0 \dots \exprB_i \dots, \stateS_2)$ and $(\expr_i, \stateS_1) \step (\expr'_i \vec{\expr}, \stateS'_1)$,
then there exist an $\exprB'_i$, $\vec{\exprB}$ and $\stateS'_2$ such that:
  \begin{itemize}
  \item $(\exprB_i, \stateS_2) \step (\exprB'_i \vec{\exprB}, \stateS'_2)$;
  \item $(\expr_0 \dots \expr'_i \vec{\expr}\dots, \stateS'_1) \RR
    (\exprB_0 \dots \exprB'_i \vec{\exprB} \dots, \stateS'_2)$.
  \end{itemize}
\end{enumerate}
\end{lemma}

By the above lemma, we now know that $\RR$ has all the properties of a strong low-bisimulation (\cf \cite[Definition 6]{prob-ni}), short of being a partial equivalence relation.
Since $\RR$ is not transitive, we consider its transitive closure $\RRR$, and verify that all the properties of a strong low-bisimulation hold for $\RRR$.

\begin{theorem}
\label{thm:strong-low-bisim}
The relation $\RRR$ is \emph{a strong low-bisimulation} on configurations.
\end{theorem}
% \begin{proof}
%   W.l.o.g. assume the size of the thread pools is 1 and the step in the consideration does not fork off any threads (the more general case is proven similarly).
%   We have a chain $(\expr_1, \stateS_1) {\RR} (\expr_2, \theta_2) {\RR} \dots {\RR} (\expr_n, \theta_n) {\RR} (\exprB, \stateS_2)$ and $(\expr_1, \stateS_1) \step (\expr'_1, \stateS'_1)$.
%   The following holds:
%   \begin{enumerate}
%   \item because $(\expr_1, \stateS_1) \RR (\expr_1, \stateS_1)$ and \Cref{lem:dwp_rel_step}:
%     $(\expr'_1, \stateS'_1) \RR (\expr'_1, \stateS'_1)$;
%   \item because $(\expr_1, \stateS_1) \RR (\expr_2, \theta_2)$ and \Cref{lem:dwp_rel_simul}
%     there is a configuration $(\expr'_2, \theta'_2)$ such that:
%     $(\expr'_1, \stateS'_1) \RR (\expr'_2, \theta'_2)$.
%   \end{enumerate}
%   Hence, $(\expr'_1, \stateS'_1) {\RR} (\expr'_2, \theta'_2)$.
%   By induction we construct configurations $(\expr'_3, \theta'_3) \dots (\expr'_n, \theta'_n), (\exprB', \stateS'_2)$ and a chain
%   \[
%     (\expr'_1, \stateS'_1) {\RR} (\expr'_2, \theta'_2) {\RR} \dots
%     {\RR} (\expr'_n, \theta'_n) {\RR} (\exprB', \stateS'_2).
%   \]
% \end{proof}
%
The theorem \Cref{thm:strong-low-bisim} in combination with \Cref{prop:dwp_lift_bisim} implies the soundness of \thelogic (\Cref{thm:adequacy}).

%%% Local Variables:
%%% mode: latex
%%% TeX-master: "main"
%%% End: